\documentclass[11pt,a4paper]{article}
\usepackage[T1]{fontenc}

\usepackage[right=3cm, a4paper]{geometry}

\usepackage[english]{babel}

\usepackage[a-2u]{pdfx}

\usepackage{lmodern}

\usepackage{amsmath}        % extensions for typesetting of math
\usepackage{amsfonts}       % math fonts
\usepackage{amsthm}         % theorems, definitions, etc.
\usepackage{amssymb}
\usepackage{bbding}         % various symbols (squares, asterisks, scissors, ...)
\usepackage{bm}             % boldface symbols (\bm)
\usepackage{bbm}

\usepackage{graphicx}       % embedding of pictures
\usepackage{fancyvrb}       % improved verbatim environment
\usepackage[
style=alphabetic,
sorting=nyt,
maxcitenames=99,
mincitenames=99,
minbibnames=99,
maxbibnames=99,
minalphanames=5,
maxalphanames=10,
backend=biber
]{biblatex} %Imports biblatex package
\usepackage[nottoc]{tocbibind} % makes sure that bibliography and the lists
\usepackage{dcolumn}        % improved alignment of table columns
\usepackage{booktabs}       % improved horizontal lines in tables
\usepackage{paralist}       % improved enumerate and itemize
\usepackage{enumitem}

\usepackage{csquotes}
\usepackage{dirtytalk}
\usepackage{relsize}
\usepackage{mathtools}
\usepackage{xfrac}
\usepackage{tikz}
\usepackage{tikz-cd}
\usepackage{exscale}
\usepackage{xcolor}         % typesetting in color
\usepackage{comment} 
\usepackage{titlesec}

\usepackage{quiver}

\mathtoolsset{showonlyrefs}  %This only shows equation numbers for equations that are referred to.

\makeatletter
\def\@makechapterhead#1{
  {\parindent \z@ \raggedright \normalfont
   \Huge\bfseries \thechapter. #1
   \par\nobreak
   \vskip 20\p@
}}
\def\@makeschapterhead#1{
  {\parindent \z@ \raggedright \normalfont
   \Huge\bfseries #1
   \par\nobreak
   \vskip 20\p@
}}
\makeatother

\newtheoremstyle{theoremnew}% name
  {.5\baselineskip\@plus.2\baselineskip\@minus.2\baselineskip}% Space above
  {.5\baselineskip\@plus.2\baselineskip\@minus.2\baselineskip}% Space below
  {\slshape}% Body font
  {}%Indent amount (empty = no indent, \parindent = para indent)
  {\bfseries}%  Thm head font
  {.}%       Punctuation after thm head
  { }%      Space after thm head: " " = normal interword space;
  {}%       Thm head spec

\theoremstyle{plain}
\newtheorem{theorem}{Theorem}[section]
\newtheorem{prop}[theorem]{Proposition}
\newtheorem{lemma}[theorem]{Lemma}

\newtheorem{cor}[theorem]{Corollary}
\newtheorem{remark}[theorem]{Remark}

\theoremstyle{remark}

\theoremstyle{definition}
\newtheorem{newdefn}[theorem]{Definition} 
\newenvironment{definition}
{\renewcommand{\qedsymbol}{$\blacktriangle$}%
\pushQED{\qed}\begin{newdefn}}
{\popQED\end{newdefn}}
\newtheorem{example}[theorem]{Example}

\renewcommand\qedsymbol{$\blacksquare$}

\definecolor{hyperlink}{RGB}{11,0,128}
\hypersetup{colorlinks, allcolors=hyperlink, linktoc=all}

\makeatletter
\newdimen\scalemath@axis
\newcommand*{\scalemath}[3]{%
  #1{%
    \mathpalette{\scalemath@aux{#2}}{#3}%
  }%
}
\newcommand*{\scalemath@aux}[3]{%
  \begingroup
    \everyvbox{}%
    \settoheight\scalemath@axis{$#2\vcenter{}$}%
    \raisebox{\scalemath@axis}{%
      \scalebox{#1}{%
        \raisebox{-\scalemath@axis}{%
          $\m@th#2#3$%
        }%
      }%
    }%
  \endgroup
}
\makeatother 

\newcommand{\email}[1]{\href{mailto:#1}{#1}}

\definecolor{MyRed}{rgb}{0.75, 0.00, 0.2}
\definecolor{MyBlue}{rgb}{0.4, 0.1, 0.70}
\definecolor{JanColor}{rgb}{0.07, 0.50, 0.00}
\definecolor{MartinColor}{rgb}{0.00, 0.31, 0.75}
\definecolor{BranoColor}{rgb}{7.00, 0.60, 0.00}
\definecolor{Gray}{HTML}{888888}

         \DeclareMathAlphabet{\mathscr}{U}{BOONDOX-cal}{m}{n}
         \SetMathAlphabet{\mathscr}{bold}{U}{BOONDOX-cal}{b}{n}
         \DeclareMathAlphabet{\mathbscr} {U}{BOONDOX-cal}{b}{n}
         \DeclareMathAlphabet{\mathpzc}{OT1}{pzc}{m}{it}

\DeclareMathAlphabet\mathbfcal{OMS}{cmsy}{b}{n}
\newcommand{\mathcat}[1]{\mathsf{#1}}

  \newcommand{\restr}[2]{{% we make the whole thing an ordinary symbol
  \left.\kern-\nulldelimiterspace % automatically resize the bar with \right
  #1 % the function
  \vphantom{\big|} % pretend it's a little taller at normal size
  \right|_{#2} % this is the delimiter
  }}
  
  \newcommand{\restrsmall}[2]{{% we make the whole thing an ordinary symbol
  \left.\kern-\nulldelimiterspace % automatically resize the bar with \right
  #1 % the function
  \vphantom{|} % pretend it's a little taller at normal size
  \right|_{#2} % this is the delimiter
  }}

\newcommand{\F}[1]{\mathcal{F}  #1 }
\newcommand{\Fun}{\mathcal{F}}

\newcommand{\qGF}{q^\mathrm{GF}}
\newcommand{\QGF}{Q^\mathrm{GF}}

\newcommand{\K}{\mathcal{K}}

\newcommand{\Z}{\mathbb{Z}}

\newcommand{\QOSC}{\mathcat{LinQSymp_{-1}}}

\newcommand{\BV}{\Delta}

\newcommand{\dhalf}{\operatorname{d}^{\frac12} \!}
\newcommand{\Ber}[1]{\operatorname{Ber} \left(#1\right)}
\newcommand{\Sfree}{{S_{\mathrm{free}}}}
\newcommand{\Sfr}[1]{{S_{\mathrm{free}}^{#1}}}

\newcommand{\Sint}{S_{\mathrm{int}}}

\renewcommand{\d}{\operatorname{d} \!}
\newcommand{\dif}{\operatorname{d} }

\newcommand{\proj}{\operatorname{proj}}

\newcommand{\Lie}{\mathfrak{L}}

\renewcommand{\Im}{\operatorname{Im}}
\renewcommand{\ker}{\operatorname{Ker}}
\newcommand{\Ker}{\operatorname{Ker}}

\newcommand{\define}{ \coloneqq }

\renewcommand{\deg}[1]{ \left\vert #1 \right\vert }
\newcommand{\after}{\circ}

\newcommand{\flip}[1]{\overline{#1}}

\renewcommand{\tilde}[1]{\widetilde{#1}}

\newcommand{\Sym}{\operatorname{Sym}}

\newcommand{\HalfDens}{\operatorname{Dens}^{\frac12}}

\title{Homotopies in Batalin-Vilkovisky Formalism}

\author{Branislav Jurčo$^{a}$, Ján Pulmann$^{b}$, and Martin Zika$^{c}$\\
\\
         Mathematical Institute, Faculty of Mathematics and Physics, Charles University, \\
        Prague, Czech Republic \\ \\
        \small $^{a}$ \email{Branislav.Jurco@mff.cuni.cz}, $^{b}$ \email{Jan.Pulmann@gmail.com}, $^{c}$ \email{martin.zika.mail@gmail.com}
}

\date{} %leave blank

\begin{document}

\maketitle

\begin{abstract}
    We review the notion of homotopy of quantum master actions in geometric Batalin-Vilkovisky formalism. Then we construct new examples of such homotopies, coming from renormalization group flow and non-infinitesimal changes of gauge fixing.  Finally, we use the field redefinitions given by these homotopies to construct spans of quantum master actions with isomorphic effective actions.
\end{abstract}

\tableofcontents

\section{Introduction}
Batalin-Vilkovisky (BV) formalism was introduced as a tool to handle perturbative gauge quantum field theories \cite{BatalinVilkovisky}. Afterwards, Schwarz and Khudaverdian \cite{schwarz:geometry_of_bv, khudaverdian} interpreted it in terms of odd symplectic geometry, and Ševera combined it with Weinstein's idea of the symplectic category \cite{weinstein:symplecticgeometry, severa:qosc}.

In BV formalism, a quantum field theory is given by a half-density $\rho$ on an odd symplectic manifold of BV fields, which is closed with respect to the BV operator $\Delta$. Consequently, one of  many insights of \emph{homological algebra} into BV formalism is a meaningful definition of homotopy of such theories $\rho_{1}, \rho_{2}$: namely a half-density $\varphi$ such $\rho_2 -\rho_1 = \Delta \varphi$. In physics terminology, $\rho$ is the exponential of a quantum master action $S$ and a change of $\rho$ by homotopy corresponds to a canonical transformation in the sense of BV formalism.

The goal of this paper is to explain some of the consequences of this notion of homotopy, and to give many examples of such homotopies. These will come from symplectic flows, the renormalization group flow and homological perturbation theory. Moreover, we investigate the space of BV transfers in the setting of linear finite-dimensional vector spaces, where we prove that ``strict'' transfers onto cohomology are all homotopy equivalent. Finally, we use these methods to construct a \say{fiber product} of theories with isomorphic minimal models.

In this way, we extend the ideas of Hamiltonian flows giving homotopic theories \cite{severa:qosc}, infinitesimal deformations of Lagrangian relations \cite{mnev:simplicialBF, mnev:discreteBF, cattaneo_mnev:remarks_on_cs}, renormalization group techniques of Costello \cite{costello:renormalisation_and_bv, costelloBook} and spans of classical BV theories \cite{farahani_saemann_wolf:spans_of_Linfty}.

The geometric perspective on BV is now gaining popularity. In addition to the references above we would like to cite the recent paper of Ševera \cite{severa:halfdens}, which, among other things, sorts out the issues regarding orientations and distributions; as well as \cite{cattaneo_mnev:BVpushforward} which seems to suggest a physical origin to the homological perturbation formulas, which we use to construct some of our homotopies. We hope this work supports the spread of these beautiful geometric ideas underpinning quantum field theory.

\paragraph{Content of the paper.}
We start by giving a quick overview of BV formalism
in Section \ref{sec:background}. In Section \ref{sec:linearBV}, we review the description of linear perturbative BV formalism in terms of homological perturbation theory.

In Section \ref{sec:ergf}, we show that a version of Costello's renormalization group flow can be seen as giving homotopic actions. In Section \ref{section:parametrizing-gauge-fixings.}, we characterize the space of gauge-fixings in linear BV formalism, proving that it's contractible if we transfer onto cohomology. In Section \ref{sec:spans} given two theories with isomorphic minimal models, we construct them as reductions of a single, bigger theory.

We end with an appendix explaining our conventions for graded geometry and containing some calculations related to suspensions of symplectic flows.

\paragraph{Note on rigor.}
We tried to write this paper for a mixed audience of mathematical and theoretical physicists. For readers worried about rigor, let us note that we are sloppy about $\mathbb Z$-graded geometry in the background Section \ref{sec:background} and the last section on spans Section \ref{sec:spans}, and functional analysis for infinite-dimensional heat kernels in Section \ref{sec:ergf}. The rest of the claims are mathematical propositions.

\begin{comment}
\begin{itemize}
    \item We don't have to introduce BV in the introduction much, as we give a very general overview in the next section
    \item We should rather explain that in BV: propagators are given by gauge fixings, which are given by homotopy operators, and how HPL fits here
    \item What we do new relative to Costello and Cattaneo-Mnev
    \item A note on rigor: the section 2 and maybe 4 is physicsy, issues with Z graded manifolds and maybe functional analysis is section 4. Section 5, 6 is rigorious (theorems).
\end{itemize}
\end{comment}

\subsection*{Acknowledgements}
We would like to thank Roman Golovko for discussions about Lagrangian cobordisms and C. Saemann and M. Wolf for discussion about spans. J.P. would also like to thank Pavol Ševera for discussions about homotopies and degrees in the context of BV pushforward. J.P. is also grateful to Alex Arvanitakis, Jian Qiu, Fridrich Valach and Michal Vorobel for further discussions. We would also like to thank the referee for useful comments and references.

The work of J.P. was supported by GAČR grant PIF-25-17640I.  B.J. thanks for support to GAČR grant 24-10031K. M.Z. was supported by the GAUK 283723 and PRIMUS/25/SCI/018 grants.

\paragraph{AI use.} We had ChatGPT Astra generate a referee report on the first arxiv version of this paper, and implemented some of its suggestions. Most notably, it pointed out that we only prove the commutativity of \eqref{diag:spanpentagon} when applied to the exponentiated actions; and found many typos.

\section{Background on Batalin-Vilkovisky formalism}\label{sec:background}

In this section, we review BV formalism, emphasizing a geometric perspective. We begin with a lightning overview of the BV fiber integral and explain how it implies the usual ``gauge invariance'' properties of the BV fiber integral. Then we expand on the details about half-densities and homotopies in BV formalism.  We formulate directly the \emph{quantum} BV formalism, the classical BV formalism is the $\hbar \to 0$, tree level approximation to the perturbative fiber integral.

Most of the material in this section is known to experts, see for example \cite{severa:qosc, costello:renormalisation_and_bv, MikhailovSchwarz2017, cattaneo_mnev_reshetikhin:perturbative_gauge_theories, GetzlerPohorence, , costelloBook, cattaneo_mnev:BVpushforward, severa:halfdens}, we also tried to give more granular references throughout the section. However, we believe that some aspects of this geometric perspective are not widely appreciated. From textbook treatments,  we recommend especially \cite[Ch.~4]{mnev:book}, but see also the reviews \cite{qiu_zabzine:intro_to_bv, Cattaneo2023BVsemidensities, cattaneo_mnev_schiavinna:Enc}. 

\subsection{Geometric Batalin-Vilkovisky formalism: an overview}\label{sec:overview}
BV formalism is a framework for extracting (formal) claims about (infinite-dimensional) path integrals. These integrals range over the BV space of fields $M$, which is a $\mathbb Z$-graded manifold equipped with a symplectic form of degree $-1$. There are standard ways to obtain such BV spaces: For example, starting from the BRST formulation of a theory one adds antifields to fields and ghosts \cite[Sec.~15.9]{weinberg}; another option is the AKSZ construction \cite{aksz}, which defines the BV space as consisting of fields of a graded $\sigma$-model with a dg shifted symplectic target. 
\smallskip

A physical theory is specified by a half-density\footnote{Half-densities are sections of the square root of the density line bundle over $M$, i.e.\ in local coordinates they can be written as
$ \rho(x) \sqrt{dx} $ where the symbol $ \sqrt{dx}$ transforms with the square root of the Jacobian
under a change of coordinates; see Section \ref{sec:halfdens}.} on $M$ which should be interpreted as
\[ \rho = e^{S(x)/\hbar} \sqrt{dx}  \in \HalfDens (M), \]
i.e.\ a combination of the action functional $S$ and a square root of the (formal) measure on $M$. There is a canonical differential $\Delta$ which acts on half-densities, known as the BV operator or BV Laplacian. In local coordinates where $\omega = \omega_{ij} dx^i dx^j$ with $\omega_{ij}$ constant, $\Delta$ is 
\begin{equation}\label{eq:explicitBV}
	\Delta (\rho(x) \sqrt{dx}) = \tfrac 12 \omega^{ij}\frac{\partial^2 \rho(x)}{\partial x^i \partial x^j} \sqrt{dx}.
\end{equation}

The main operation of the BV formalism is the BV fiber integral or BV pushforward. Given a (surjective) \emph{Lagrangian relation} $R\colon M \to N$, the BV fiber integral is a map \cite[Section~6]{severa:qosc} \cite[Remark~2.14]{cattaneo_mnev_reshetikhin:perturbative_gauge_theories} 
\[ \int_R \colon \HalfDens(M) \to \HalfDens(N). \]
Very roughly, this should be understood as integrating over the fibers of the relation $R$
\begin{equation}
\left(\int_R \rho\right) (n) = \mkern-20mu \int\limits_{m \in M \, \mathrm{ s.t.} \,(m, n)\in R}\mkern-30mu \rho(m)\,.
\end{equation}
We explain in  Section \ref{sec:BVFiberIntegral} how to take the half-density factors into account.
\medskip

The fundamental property of the BV fiber integral is that it is a chain map
\begin{equation} \label{eq:BVischainmap}
	\int_R \Delta_M \rho =  \Delta_N \int_R\rho.
\end{equation}
Moreover, it is compatible with composition of relations: for relations $R\colon M \to N$ and $R' \colon N\to O$, one defines
$R'\circ R = \{ (m, o) \mid \exists n \in N \text{ s.t. } (m, n)\in R \text{ and } (n, o) \in R' \}\subset M\times O$; then\footnote{Note that in general, the composition of smooth relations is not necessarily a smooth relation anymore.}
\begin{equation}\label{eq:BVFubini}
	\int_{R'}\int_R \rho =   \int_{R'\circ R}\rho.
\end{equation}
These two properties should be thought of as the Stokes and Fubini theorems in the BV formalism. 
\medskip

\textbf{Let us now relate the above description with the usual presentations of the BV formalism.} 
A simple example of a Lagrangian relation is given by a Lagrangian submanifold $L\subset M$, which can be seen as a relation between $M$ and the point. Half-densities on a point are just numbers, the BV operator is zero, and we get
\begin{equation}\label{eq:integral_of_exact_zero} \int_L \Delta_M \rho = \Delta_\text{pt}\int_L \rho = 0. \end{equation}
\begin{equation}\label{eqtext:exact}
\text{\textit{That is, integrals of $\Delta$-exact functionals over Lagrangians vanish.}} \end{equation}

Let us consider a family of Lagrangian submanifolds $L_t \subset M$, and moreover assume they combine into a Lagrangian submanifold\footnote{Such Lagrangians are called \emph{Lagrangian cobordisms} by symplectic topologists \cite{Arnold1980}. The space $T^*[-1]\mathbb R$ has a degree $0$ coordinate $t$ and a degree $-1$ (fiber) coordinate $\tau$. The BV operator is $\partial^{2}/\partial t \partial \tau$. $T^*[-1]\mathbb R$ is the natural interval object in BV formalism. The standard naming conventions of BV would dictate that $\tau$ should be called the \emph{antitime}.} 
\begin{equation} \label{eq:LagrCob1} \tilde{L} = \left\{ (l_t, t, g(t, l_t)) \mid l_t \in L_t \right\}\subset \overline{M} \times T^*[-1]\mathbb R \end{equation}
 with some function $g$ on $\mathbb R \times M$. Then we can perform the BV fiber integral of a half-density on $M$ to get a half-density on\footnote{For $R\colon M \to N$, the degree of the map $\rho \mapsto \int_R \rho$ is \[ \left| \int_R - \right| = \sum_i i \dim R_i - \dim N /2.\]} $T^*[-1]\mathbb R$
 \[ \int_{\tilde{L}}\rho := \left(\mu(t) + \nu(t)\tau \right)\sqrt{dtd\tau}.\]
We can compute $\nu(t)$ by Berezin integral over the fiber $T_t^*[-1]\mathbb R$ at $t$, by \eqref{eq:BVFubini} it is \begin{equation} \label{eq:nutcomponentLt} \nu(t) = \int_{T_t^*[-1]\mathbb R} \int_{\tilde L} \rho  = \int_{L_t} \rho.\end{equation}
If $\rho$ is $\Delta$-closed, then so is $\int_{\tilde{L}}\rho$. Since $\Delta = \partial^{2}/\partial t \partial \tau$, we get
\begin{equation}
    0=\Delta \int_{\tilde{L}}\rho = \partial_t \partial_\tau \left( \mu(t) + \nu(t) \tau \right)\sqrt{dt d\tau} = \partial_t \nu(t) \sqrt{dt d\tau},
\end{equation}     
i.e.\ $\nu(t)$ is constant.
\begin{equation}\label{eqtext:closed}\parbox{0.8\textwidth}{\textit{In other words,
 for closed integrands $\rho$, the integral  $\int_{L}\rho$ is unchanged under small deformations of the Lagrangian $L$ which extend to a family parametrized by $T^*[-1]\mathbb R$.}}\end{equation}
See Section \ref{sec:forms} for comments on when such $\widetilde{L}$ exist, and Section \ref{sec:Hamflows} for an explicit example of such $\widetilde{L}$ coming from a Hamiltonian flow, giving an interpretation to $g(t)$ and $\mu(t)$.

 \subsection{Half-densities and fiber integrals} \label{sec:halfdens} Let us now supply more details on half-densities. They 
 are written in local coordinates as $\rho_x(x)\sqrt{dx}$. If $\phi$ is a  diffeomorphism expressed in local coordinates by $y(x) = \phi^*(y)$, then the defining property of a half-density is that it pulls back via\footnote{This implies the following patching formula for half-densities: if $\rho$ pulls back to a coordinate patch $U_x$ to $\rho_x(x)\sqrt{dx}$ and similarly on $U_y$, then on $U_x \cap U_y$ the coefficient functions are related by \[ \rho_x(x) = \rho_y(y(x)) \left| \operatorname{Ber}{ \frac{\partial y(x)}{\partial x}}\right|^{\tfrac{1}{2}}.\]} 
 \begin{equation}\label{eq:berezinianpullback}
     \phi^*\left(\rho_y(y) \sqrt{dy}\right) = \rho_y(y(x)) \left| \operatorname{Ber}\frac{\partial y}{\partial x}\right|^{\tfrac 12} \sqrt{dx}.
 \end{equation}
Here $\operatorname{Ber}$, the Berezinian of an even matrix $A$, is given by e.g.\footnote{We only use the underlying $\mathbb Z/2$-grading in this definition; for $\mathbb Z$-graded manifolds the Berezinian has degree $0$.}
 \begin{equation}\label{eq:Berezinianexplicit}
    \operatorname{Ber}\begin{pmatrix} A_{\text{even}\to \text{even}} & A_{\text{odd}\to \text{even}} \\ A_{\text{even}\to \text{odd}} & A_{\text{odd}\to \text{odd}}\end{pmatrix} 
    = \frac{\det(A_{\text{even}\to \text{even}} - A_{\text{odd}\to \text{even}} A_{\text{odd}\to \text{odd}}^{-1} A_{\text{even}\to \text{odd}})}{\det(A_{\text{odd}\to \text{odd}})}
\end{equation}
 and satisfies
 \begin{equation}\label{eq:Berezinianstr}
      \operatorname{Ber}(e^{A}) = e^{\operatorname{str}{A}}.
 \end{equation}
 Even though we wrote the BV operator using local coordinates in \eqref{eq:explicitBV}, it is defined on half-densities on an odd symplectic manifold without any additional choices \cite[Sec.~2]{khudaverdian}, \cite{severa:origin_of_bv}. 
 If we choose a nowhere-vanishing $\Delta$-closed half-density $\rho_{0}$, we can induce a BV operator $\Delta_{0}$ on functions on $M$ by
 \begin{equation}\label{eq:Delta0}
      (\Delta_{0}F) \rho_{0} = \Delta(F \rho_{0}).
 \end{equation}
 This endows the algebra of functions on $M$ with a BV operator $\Delta_0$ in the usual sense,\footnote{This is the BV algebra of observables
 which appears often in literature, and it depends on a choice of a half-density. If $M$ is a graded vector space, all translation-invariant
 half-densities induce the same BV operator (since they differ by a constant); thus we get a canonical BV operator on functions in this case.
 Half-densities are a more natural object for integration, but they don't form an algebra; together with the Gerstenhaber algebra of functions they form a \emph{Tamarkin--Tsygan calculus} \cite[Def.~3.2.2]{TamarkinTsygan2005}.}
 i.e.\ $\Delta_{0}$ is a degree 1 \emph{second order differential operator}, squares to zero, $\Delta_{0}(1) = 0$ and it is related to the 1-shifted Poisson bracket on functions via
 \begin{equation}\label{eq:Deltasecondorder}
      \Delta_{0}(FG) = \Delta_{0}F G + (-1)^{\deg{F}} F \Delta_{0} G + (-1)^{\deg{F}} \{ F, G\}.
 \end{equation} 
 For a function $H$ on $M$, we have for Hamiltonian flows of half-densities the following very useful identity for the Lie derivative \cite[Eq.~2.16]{khudaverdian}
\cite[Thm.~2.3]{severa:qosc} 
\begin{equation}\label{eq:CartanMagicFormula}
      \Lie_{\{H, -\}} \rho = (-1)^{|H|}[\Delta, H\cdot]\rho, \quad \text{where $H\cdot$ is the operator of multiplying  by $H$}.
 \end{equation} 
 Applying this to the reference $\Delta$-closed half-density $\rho_{0}$ gives the usual definition of $\Delta_{0}$ as the 
divergence of Hamiltonian vector fields
\begin{equation}\label{eq:DivergenceBV}
    \Lie_{\{H, -\}} \rho_{0} = (-1)^{|H|}\Delta_{0}(H) \rho_{0}.
\end{equation}
From \eqref{eq:CartanMagicFormula}, we also get that $\Delta$ commutes with Hamiltonian flows, as by $\Delta^2 = 0$ 
\[ [ \Delta, \Lie_{\{H, -\}} ] = (-1)^{|H|}  [\Delta, [\Delta, H]] = 0.\]

\subsubsection{Physical interpretation}
In physics, the integrands we want to consider are $\rho=e^{S/\hbar}\sqrt{dx}$ (for computing the partition function) and $\rho = F e^{S/\hbar}\sqrt{dx}$ (for computing 
 the expectation value of the observable $F$).  These half-densities are $\Delta$-closed if 
\begin{equation}\label{eq:QME}
\hbar \Delta_0 S + \frac 12 \{S, S\} = 0 \; \text{ and } \; \hbar \Delta_0 F + \{S, F\} = 0, \; \text{respectively, with $\rho_0 =\sqrt{dx}$}.
\end{equation} 
The first equation is the \emph{quantum master equation}, and the second equation tells us that $F$ is a cocycle with respect to the 
twisted BV operator $\hbar \Delta_0  + \{S, -\}$. By \eqref{eq:integral_of_exact_zero}, we see that only the \emph{twisted BV cohomology class of} $F$ is measured by the path integral. 

 \smallskip
 
 The two emphasized statements \eqref{eqtext:exact} and \eqref{eqtext:closed} in Section \ref{sec:overview}, about integrals of $\Delta$-exact and $\Delta$-closed half-densities, are a generalization of gauge-invariance. In the BV formulation of gauge theories, gauge-invariant
 observables are $\Delta$-closed, $\Delta$-exact observables are pure gauge and changes of gauge fixing are encoded in changing the Lagrangian $L$.
 \smallskip

The quantum master equation \eqref{eq:QME} is therefore a fundamental property of a theory. It ensures that expectation values, given by integrals over $L\subset M$, are well defined. If we instead perform a BV fiber integral over a relation $R \colon M \to N$, the chain map property \eqref{eq:BVischainmap} tells us that $\int_R\rho$ is also closed if $\rho$ is closed; or that the effective action $W$ defined as \cite[Eq.~(8)]{KrotovLosev}
\[ e^{W/\hbar} \sqrt{dy}:= \int_R e^{S/\hbar} \sqrt{dx}  \]
satisfies the quantum master equation if $S$ does so.

\subsubsection{Odd Fourier transform}
Another way to approach half-densities is by identifying them with differential forms as follows \cite[Sec.~3]{khudaverdian}. If $M = T^{*}[-1]X$ where $X$
is an oriented manifold concentrated in even degrees, then there is a canonical isomorphism between half-densities and (inhomogeneous) forms of degree $\dim X$ 
\begin{equation}
    \HalfDens(T^{*}[-1]X) \cong \Omega^{\bullet} X,
\end{equation}
which in local coordinates $q^{i}$ on $X$ (and $\pi_{i}$ for the $T^*[-1]$-fiber coordinates) is given by
\begin{equation}\label{eq:OFT}
    \sum_{I \text{ multi-index}} \mkern-24mu f^{I}(q) \pi_{I} \sqrt{dq d\pi} \mapsto \pm\sum_I f^{I}(q) dq^{\widehat{I}}  \text{ where ${\widehat{I}}$ is the complement to $I$}
\end{equation}
where the sign is the product of Koszul signs for rearranging $\pi_I \pi_{\widehat{I}}$ to $\pi_1\dots \pi_{\dim{X}}$ and for reversing of $\pi_I$ \cite[Eq.~65]{albert_bleile_frohlich:bv_integrals}. 

Under this isomorphism, the operator $\Delta$ corresponds exactly to the de Rham differential \cite{witten:antibracket}, and multiplication of half-densities with functions on $M=T^*[-1]X$
corresponds to contracting with polyvector fields on $X$ under $\mathcal O(T^{*}[-1]X) \cong \Gamma (\Lambda^{\bullet} TX)$. The relation \eqref{eq:CartanMagicFormula} becomes the Cartan magic formula.
Since the isomorphism \eqref{eq:OFT} can be written as \cite[Eq.~18]{schwarz:geometry_of_bv}
\begin{equation}
    \rho \mapsto \widehat{\rho} = \int e^{\pi_{i} dx^{i}} \rho(x, \pi) d\pi \quad \text{(the Berezin integral)}
\end{equation}
it is called the \emph{odd Fourier transform} \cite{schwarz:geometry_of_bv}.

\begin{example}Let us make explicit the odd Fourier transform for $T^*[-1]\mathbb R$. Let $t$ and $s$ be two different coordinates 
on $\mathbb R$ of the same orientation, and $\tau$, $\sigma$ the respective fiber coordinates; then $\tau$ transforms as $\partial/\partial{t}$ which implies
$\tau = \tfrac{\partial s}{\partial t} \sigma$. The Berezinian of the coordinate transform is
\begin{equation}
	\operatorname{Ber} \begin{pmatrix}
	    \tfrac{\partial t}{\partial s} &    \tfrac{\partial t}{\partial \sigma} \\[3pt]
           \tfrac{\partial \tau}{\partial s} &    \tfrac{\partial \tau}{\partial \sigma}
	\end{pmatrix}   = \operatorname{Ber}\begin{pmatrix} \tfrac{\partial t}{\partial s} & 0 \\[3pt]
	   - \tfrac{\partial^2 t}{\partial s^2} \left(\tfrac{\partial s}{\partial t} \right)^2 & \tfrac{\partial s}{\partial t} \end{pmatrix} \stackrel{\eqref{eq:Berezinianexplicit}}{=} \left(\frac{\partial t}{\partial s}\right)^2.
\end{equation}
A half-density $(\mu(t)+\nu(t)\tau)\sqrt{dt d\tau}$ transforms to
\[ \left(\mu(t(s)) + \nu(t(s)) \frac{\partial s}{\partial t} \sigma \right) \sqrt{\operatorname{Ber} \left( \frac{\partial (t, \tau)}{\partial (s, \sigma)}\right) } \sqrt{ds d\sigma} = \left(\frac{\partial t}{\partial s} \mu(t(s)) + \nu(t(s))\sigma \right) \sqrt{ds d\sigma}. \]
Under the odd Fourier transform \eqref{eq:OFT}, these half-densities are sent to
\begin{equation}
	\mu(t)dt + \nu(t) \quad \text{and} \quad  \mu(t(s)) \frac{\partial t}{\partial s} ds + \nu(t(s))
\end{equation}
which is consistent, i.e.\ the odd Fourier transform of a half-density is independent of coordinates. We see why we need to choose an orientation on $X$: allowing orientation-reversing change of coordinates would mean that $\widehat\rho$ is only defined up to a sign.
\end{example}

\subsubsection{Forms and BV formalism} \label{sec:forms}
It is common to parametrize choices in BV formalism by extending to differential forms on some parameter space, for example $\mathbb R$, or the interval, or higher simplices \cite{costello:renormalisation_and_bv, GetzlerPohorence} and \cite[Ch~5.10]{costelloBook}. These forms should be understood as odd Fourier transforms of half-densities on $T^*[-1]$ of this parameter space; the de Rham differential being the BV operator $\Delta$. This interpretation of the de Rham differential is useful for homological perturbations.

 For example, BD algebras linear over $\Omega(\Delta^k)$ of Costello and Gwilliam \cite[Sec~8.2.2]{costello_gwilliam:fact_II} are closely related to the homotopies in Section \ref{sec:ergf}, see Section \ref{sec:zucchini}.

Similarly in \cite[Sect~2]{MikhailovSchwarz2017}, the space of parameters is denoted $\Lambda\owns \lambda$, and one considers a $\lambda$-dependent (exponentiated) action $\exp S(x, \lambda)\sqrt{dx}$ as extending to a form $\exp[S + B_a d \lambda^a]\sqrt{dx}$. This is the odd Fourier transform of the half-density
\[ e^S\prod_a (B_a + l_a) \sqrt{dx d\lambda dl} \in \HalfDens(M\times T^*[-1]\Lambda), \]
where $x$ are coordinates on a $(-1)$-shifted symplectic manifold $M$, while $l_a$ are the antifields to $\lambda$ on $T^*[-1]\Lambda$. Similarly, the form on the space $LAG$ in \cite[Sec.~3]{MikhailovSchwarz2017} should be seen as a half-density on $T^*[-1]LAG$, and similarly in \cite{GetzlerPohorence}. The integrals \cite[eq.~(14)]{MikhailovSchwarz2017} and \cite[Thm.~4.1]{GetzlerPohorence} should be seen as BV fiber integrals, see \cite{MikhailovFamilies, severa:halfdens} for more details.

Extending a family of Lagrangians $(L_{\lambda}\subset M)_{\lambda \in \Lambda}$ to a Lagrangian in $M\times T^*[-1]\Lambda$ is a cohomological problem, which is discussed in \cite{schwarz:geometry_of_bv, GetzlerPohorence, severa:halfdens}.

\subsubsection{BV fiber integral} \label{sec:BVFiberIntegral}
With more knowledge of half-densities, let us now go back and define the BV fiber integral along a surjective Lagrangian relation $R \colon M \to N$. Recall that a Lagrangian relation between symplectic spaces $M\to N$ is a Lagrangian subspace of $\overline{M}\times N$. Here, $\overline{M}$ denotes $M$ with the symplectic form multiplied by $-1$. A (Lagrangian) relation $M\to N$ is surjective if each point $n\in N$ is contained in some pair $(m, n)\in R$. 
\medskip

Each coisotropic subspace $C\subset M$ induces such surjective Lagrangian relation, the \emph{coisotropic reduction} $M \to C/C^\omega$ given by $\{(c, c\mod I) \mid c \in C\}$. Here, $I\subset TM$ is the isotropic distribution defined on $C$, given by $I_c \equiv(T_c C)^\omega\subset T_c C \subset T_c M$. Conversely, a surjective Lagrangian $R\subset \overline{M}\times N$ is given by
\[ \{ (c, \phi(c) \mid c \in C \subset M \} \]
where $C = \operatorname{proj}_M R\subset M$ is the domain of definition of $R$ and $\phi \colon C \to N$ is a surjective map, which we will assume is a submersion. The subspace $C\subset M$ is coisotropic\footnote{$C$ is coisotropic since it is the composition of coisotropic relations $\text{pt} \xrightarrow{N} N \xrightarrow{R^T} M$. For $R\subset M \times N$, its transpose $R^T$ is the relation $\{(n, m) \mid (n, m)\in R\} \subset N\times M$.}  and the isotropic distribution $I_c = (T_c C)^\omega$ spans the fibers of $\phi$. By choosing suitable complements to $I \subset T_cC \subset T_c M$, we can decompose the tangent space as 
\[ T_c M \,\cong \,\underbrace{T_{\phi(c)} N \oplus I_c}_{T_c C} \oplus \, T_c M/T_c C \,\cong\,  T_{\phi(c)} N \oplus T^*[-1] I_c \]
where we used the symplectic pairing on $T_c M$ to identify $T_cM/T_c C$ with the shifted dual of $I_c$. A half-density on $T_c M$ therefore induces a half-density on $T_{\phi(c)}N$  and a \emph{density} on $I_c$; these are independent of the choices involved in the decomposition.\footnote{The density factor comes from the fact that if a basis of $I$ transforms with a matrix $A$, the basis of the fiber $I^*[-1]$ transforms with the inverse of $A$. The total Berezinian of the transformation is \[ \Ber{\begin{matrix}
    A & 0 \\ 0 & A^{-1}
\end{matrix}}  = \Ber{A}^2, \]which follows because both the inverse and the parity shift invert the Berezinian. This is the reason why half-densities can be integrated over Lagrangians only in \emph{odd} symplectic geometry. See \cite{khudaverdianvoronovDFOSG} and \cite[Lemma 3.5]{jpz:lagr_rel_and_quantum_Linfty}. } The BV integral $\int_R \rho$ is given by the integral of this density factor along the fibers of $\phi \colon C \to N$. 
\medskip

Non-surjective relations can also be used for transfer, but the resulting half-densities will have Dirac-delta-like singularities \cite{severa:qosc, jpz:lagr_rel_and_quantum_Linfty} or our \eqref{eq:LagrCob1}. If the singularity is along an even direction $t$, it can be equivalently expressed by the half-density $\tau \sqrt{dt d\tau}$, where $\tau$ is the \say{antifield} to $t$. See also \cite{severa:halfdens}, which also includes a careful discussion of orientations; we quietly assume our manifolds are orientable.

\subsubsection{Computing the BV fiber integral}
Let us give some practical recipes for computing BV fiber integrals.
First, one can use the \emph{odd Fourier transform}: given a half-density $\rho$ on $T^*[-1]X$, the integral of $\rho$ over the Lagrangian $X$ is equal to the integral of the differential form $\widehat{\rho}\in \Omega^\bullet(X)$ over the base $X$
    \[ \int_X \rho = \int_X \widehat{\rho}.\]
    
    One can also construct Lagrangian submanifolds of $T^*[-1]X$ as odd conormal bundles of $Y\subset X$: if we denote $C[-1]Y = \{ (\alpha, y) \mid y\in Y \text{ and } \alpha \in \operatorname{Ann}T_y Y \subset T_y^*[-1]X \}$ we have \cite[Lemma~3]{schwarz:geometry_of_bv}
    \[ \int_{C[-1]Y} \rho = \int_Y \widehat\rho.\]

	More generally, we can consider Darboux coordinates $q^i, \pi_i$ with  $q^i$ \emph{not necessarily even}. Then the integral of $\rho(q^i, \pi_i)\sqrt{dq^id\pi_i}$ over the Lagrangian given by $\pi_i = 0$ is given by (Lebesgue-Berezin) integral
	\begin{equation} \label{eq:integraloverLagrangian}\int \rho(q^i, 0) dq^i.\end{equation}

	Let us now show how to use \eqref{eq:integraloverLagrangian} to compute the change of $\int_L \rho$ under a small Hamiltonian deformation of $L$. The result will be, unsurprisingly, given by the integral of  $\Lie_{\{H, -\}}\rho$, but we perform the calculation for demonstration purposes.
	\begin{example}[Deformations of Lagrangians] \label{ex:ComputationFlow}
		 Let $H$ be an infinitesimal function, of degree $-1$, the Hamiltonian. The corresponding Hamiltonian vector field $\{H, -\}$ is of degree 0. Let us choose Darboux coordinates $x^a = (q^i, \pi_i)$ for which $L = \{ \pi_i = 0,  \forall i\}$. Then $L_H = \Phi^{\{H, -\}}(L)$, the time $1$ flow of $L$, is given by vanishing of $(\Phi^{\{H, -\}})_*(\pi_i) = \pi_i - \{H, \pi_i\}$. Let us thus introduce new Darboux coordinates $\tilde{x}^a = x^a-\{H, x^a\}$, in which we can integrate over the Lagrangian as described in \eqref{eq:integraloverLagrangian} just above this example. The Berezinian of this coordinate transformation is 
		\[ \operatorname{Ber}\frac{\partial \tilde{x}^a}{\partial x^b} = \operatorname{Ber} \left( \delta_{ab} -  \frac{\partial}{\partial x^b} \{ H, x^a \} \right) \stackrel{\eqref{eq:Berezinianstr}, \eqref{eq:Hamiltonian}}{=} 1 + \omega^{ab}\frac{\partial^2 H}{\partial x^a\partial x^b}.\]  
	 This is by \eqref{eq:explicitBV} equal to $1+2\Delta_0 H$. Denote the coefficient function $\rho_{\tilde{x}}$ of the half-density $\rho$ in coordinates $\tilde{x}$, and similarly $\rho_{x}$, which satisfy
		\[ \rho_{\tilde x} \sqrt{d\tilde x} = \rho_x \sqrt{dx} \quad \implies \quad \rho_{\tilde{x}}(\tilde{x}) = \rho_x(x(\tilde x)) \sqrt{\operatorname{Ber}\frac{\partial x}{\partial \tilde x}} = \rho_x(x(\tilde x))(1-\Delta_0 H).\]

		We can now compute the integral over the deformed Lagrangian $L_H$
		\[ \int_{L_H}\rho = \int_{\tilde x^a} \rho_{\tilde x} d\tilde x^a=\int_{\tilde x^a}\rho_x(x(\tilde x))(1-\Delta_0 H)d\tilde x^a.\] 
		Expressing $x$ in terms of $\tilde x$ gives
		\[ \rho_x(x(\tilde x)) = \rho_x(\tilde x^a + \{H, x^a\})= \rho_x(\tilde{x})+\{ H, \rho_x \}(\tilde{x}) \]
		and together 
		\[ \int_{L_H}\rho = \int_{\tilde{x}} \left(\rho_x - \rho_x \Delta_0 H + \{H, \rho_x\} \right) d\tilde{x}=\int_{\tilde{x}} (\rho_x - \Delta_0(H \rho_x)-H \Delta_0 \rho_x).\]
		Even though the integration variable is called $\tilde{x}$, this is an integral over the original Lagrangian $L$. Converting it back to an integral in terms of half-densities using \eqref{eq:Delta0} with $\rho_0 = \sqrt{dx}$ and \eqref{eq:Deltasecondorder}, we get
		\begin{equation}\label{eq:exampleresultH}\int_{L_H} \rho = \int_L \left(  \rho - \Delta (H \rho) - H \Delta \rho\right) \stackrel{\eqref{eq:CartanMagicFormula}}{=} \int_L \rho + \int_L \Lie_{\{H, -\}}\rho.\end{equation}
	\end{example}
\bigskip
    
     Let us finish with some tips on how to perform the BV fiber integral over a general relation $R\colon M \to N$. For example, we can choose a test half-density $\sigma$ on $N$ and integrate against $\int_R \rho \in \HalfDens(N)$ 
    \[ \int_N \sigma \int_R \rho = \int_{R\subset \overline{M} \times N} \sigma \otimes \rho. \]
    Let us emphasize that in this formula, the integral over $N$ is the integral of a \emph{density} on $N$, while the integral on the RHS is the integral of a half-density over a Lagrangian submanifold $R\subset \overline{M} \times N$, which can be performed as in \eqref{eq:integraloverLagrangian}. We use this in e.g.\ \eqref{eq:referenceforintegralagainstest}.
\smallskip
    
     Alternatively, we can choose a Lagrangian $L \subset N$ and use compatibility of BV fiber integrals with composition of relations
    \[ \int_L \int_R \rho = \int_{L\circ R} \rho, \]
    where $L\circ R$ is a Lagrangian submanifold of $M$ and on the right we integrate a half-density on $M$. This was used to isolate the component $\nu(t)$ of the transferred half-density on $T^*[-1]\mathbb R$ in \eqref{eq:nutcomponentLt}.

\subsubsection{Convergence of integrals} \label{sec:convergence} Of course, the convergence of the BV fiber integral is not guaranteed even for finite-dimensional $M$. One option is to consider compactly-supported or rapidly decaying half-densities \cite[Ch.~4]{costelloBook}. If $M = V$ is a vector space (or working asymptotically), we can define the integral $\int_L e^{S/\hbar}$ perturbatively if the quadratic part $\Sfree$ of the action functional is non-degenerate on $L$. This is the basis of perturbative BV formalism, where a Lagrangian $L$ is chosen in order to make $\Sfree$ non-degenerate. We will be working in this linearized set-up from Section \ref{sec:linearBV} onwards.
\smallskip

It is not always possible to find a Lagrangian $L\subset V$ on which $\Sfree$ is non-degenerate. This is parametrized by the cohomology of the differential $\{\Sfree, -\}$. This cohomology is the smallest possible target for a Lagrangian surjection from $V$ such that the BV fiber integral is perturbatively defined.

\subsection{Homotopies in Batalin-Vilkovisky formalism}
We have seen that BV integrals of exact half-densities vanish 
\eqref{eq:integral_of_exact_zero}, thus it makes sense to consider two $\Delta$-closed half-densities to be equivalent if they differ by an exact half-density. On the other hand, two half-densities on $M$ which can be interpolated by a $\Delta$-closed half-density on $M\times T^*[-1]\mathbb R$ or $M\times T^*[-1][0,1]$ should also be seen as equivalent.\footnote{We will stick to $T^*[-1]\mathbb R$ for typographical reasons, with tacit understanding that some of these half-densities might be defined only on an interval inside $\mathbb R$. 
\medskip

See also Section \ref{sec:forms} and references therein for a version of this definition parametrized by differential forms.
}
These two conditions are in fact equivalent: if 
\[ \rho_1 - \rho_0 = \Delta \varphi, \]
then 
\begin{equation}
    \label{eq:interpolationHalfDens}\rho(t) = [\tau\left((1-t) \rho_0 + t \rho_1\right) -  \varphi]\sqrt{dt d\tau}
\end{equation} 
is a $\Delta$-closed half-density on $M\times T^*[-1]\mathbb R$. Conversely, a closed half-density $ (\tau \rho_t + \varphi_t) \sqrt{dt d\tau} $ on the product $M\times T^*[-1]\mathbb R$ tells us that the time derivative of its $\tau$ component is $\Delta$-exact, since
\begin{equation}
\begin{split}
    \Delta_{M\times T^*[-1]\mathbb R} (\tau \rho_t - \varphi_t) \sqrt{dt d\tau} &= (-\tau \Delta \rho_t  - \Delta \varphi_t  + \partial_t \rho_t\ )\sqrt{dtd\tau} \\
    \text{which vanishes iff } \quad  \Delta \rho_t &= 0 \quad \text{and} \quad \partial_t \rho_t = \Delta \varphi_t.
\end{split}
\end{equation}

\begin{definition}\label{def:homotopy}
	Let us therefore call two $\Delta$-closed half-densities $\rho_{0,1}$ on $M$ \emph{homotopic} if they satisfy one of these equivalent conditions: they differ by $\Delta\varphi$; or they can be interpolated by a $\Delta$-closed family
	\begin{equation}\label{eq:halfdensinterpolation}
	   (\tau \rho_t - \varphi_t) \sqrt{dt d\tau} \in \HalfDens{(M\times T^*[-1]\mathbb R)}.
	\end{equation}	
\end{definition}

Note that a general interpolation $\tau \rho_t - \varphi_t$ doesn't need to be of the linear form \eqref{eq:interpolationHalfDens}. However, as the interpolation is $\Delta$-closed, we have that $\rho_1 - \rho_0 = \Delta \int_0^1 \varphi_t$ and one can then also write the (in general different) interpolation \eqref{eq:interpolationHalfDens} using this primitive $\int_0^1 \varphi_t$.

Under the assumption that $\rho_t$ is invertible, an existence to a lift \eqref{eq:halfdensinterpolation} is further equivalent to an existence of Hamiltonian flow taking $\rho_0$ to $\rho_1$ {\cite[Ch.~5.~Sec~10.1.]{costelloBook}}. Indeed, using \eqref{eq:CartanMagicFormula} and the $\Delta$-closedness of $\rho_t$,
\begin{equation}
	\tfrac{\partial}{\partial t} \rho_t = 
	\Delta \varphi_t =  \Delta(\varphi_t \rho_t^{-1} \rho_t) = [\Delta, \varphi\rho_t^{-1} \cdot] \rho_t
	 = -\Lie_{ \{ \varphi\rho_t^{-1}, - \} } \rho_t
\end{equation}
and conversely, the small increment of a closed half-density under Hamiltonian flow is exact, again by \eqref{eq:CartanMagicFormula}. 

In physical literature, homotopy of half-densities is usually written in terms of infinitesimal changes of the  action functional: replacing
\begin{equation}
	\label{eq:physicsycanonicaltransformation}S \leadsto S + \{S, R\} + \hbar \Delta_0 R \end{equation}
transforms the half-density $e^{S/\hbar}\sqrt{dx}$ to
\[ e^{S/\hbar}\sqrt{dx} \leadsto e^{S/\hbar}\sqrt{dx} + e^{S/\hbar} \hbar^{-1}(\{S, R\} + \hbar \Delta_0 R)\sqrt{dx} = e^{S/\hbar}\sqrt{dx} + \Delta( R e^{S/\hbar}\sqrt{dx}), \]
see e.g.\ \cite{schwarz:symmetry}.

\begin{remark}[Maurer-Cartan elements]The considerations above are of course very similar to the standard theory of Maurer-Cartan elements and their homotopy \cite{dotsenko_shadrin_valette:MC_deformation_theory}. Indeed, a quantum master action $S$ is a Maurer-Cartan element in the dgla \[
	\mathcal{O}(M)[-1], \; \hbar\Delta_0, \; \{-,-\}
	\]
	and the change as in \eqref{eq:physicsycanonicaltransformation} is an infinitesimal gauge transformation by $R$. We note that this interpretation requires a choice of a reference half-density to get the BV operator on functions on $M$.
\end{remark}

\paragraph{Homotopic actions from BV fiber integrals.} It is easy to see that homotopic actions integrate to homotopic actions. Either we use \eqref{eq:BVischainmap} to get
\[ \rho_1 - \rho_0 = \Delta \varphi \quad \implies \quad \int_R \rho_1 - \int_R \rho_0 = \Delta \int_R \varphi \]
for a reduction $R\colon M \to N$, or we consider a closed half-density on $M\times T^*[-1]\mathbb R$ as in \eqref{eq:interpolationHalfDens} and integrate along the relation $R\times \operatorname{id} \colon M\times T^*[-1]\mathbb R\to N\times T^*[-1]\mathbb R$. 
\smallskip

Instead of $T^*[-1]\mathbb R$-dependent half-densities, we can use $T^*[-1]\mathbb R$-dependent Lagrangian relations such as
\begin{equation} \label{eq:time-dependentrelation} M \to N \times T^*[-1]\mathbb{R}.\end{equation}
Then the BV fiber integral takes a closed half-density on $M$ to a closed family of half-densities on $N$, i.e.\ we get homotopic actions by transferring over a family of Lagrangians. We will now show how to construct such relations.

\subsubsection{Hamiltonian flows and gauge-fixing fermions}\label{sec:Hamflows}
%We will now look again on deformations of Lagrangians by Hamiltonian flows, generalizing and giving a more conceptual version of Example \ref{ex:ComputationFlow} in terms of homotopies. 
 In Example \ref{ex:ComputationFlow}, we considered an infinitesimal Hamiltonian $H$ and its flow carrying a Lagrangian $L\subset M$ to $L_H\subset M$. The result of the example is that $\int_{L_H}\rho - \int_{L}\rho = \int_{L} \Lie_{\{H, -\}} \rho$. For closed $\rho$ the RHS is $\Delta$-exact by the Cartan magic formula \eqref{eq:CartanMagicFormula}, and thus the integrals over $L$ and $L_H$ are homotopic. To demonstrate the use of Lagrangian cobordisms, we will now construct a $T^*[-1]\mathbb R$-parametrized interpolation between $\int_{L_H}\rho$ and $\int_{L}\rho$ not by linear interpolation \eqref{eq:interpolationHalfDens}, but directly from the Hamiltonian flow. 
 \medskip 
 
 Let $H$ be a degree  $-1$ function on $M$, not assumed to be infinitesimal anymore. Then one can construct a Lagrangian relation\footnote{See Appendix \ref{app:Suspensionproof} for a detailed construction.}
\begin{equation}\label{eq:Lagrsuspensionlift}
    \Phi^{\text{rel}}_H= \{ (\Phi^t_H(m), m, t, -H(m)) \mid m \in M, t \in \mathbb R \} \colon M \to M \times T^*[-1]\mathbb R.
\end{equation}
We can use this relation to get a parametrized half-density on $M$ (see \eqref{eq:appendixBVpushforward} for derivation)
\[ \int_{\Phi^\text{rel}_H} \rho =  (\Phi^t_H)^*(\rho) \cdot (\tau+H)\sqrt{dt d\tau}  \in \HalfDens (M\times T^*[-1]\mathbb R). \]
Let us explain how this generalizes Example \ref{ex:ComputationFlow}. Let
\begin{equation}
    \widetilde{L}_H = L\times \text{id} \, \circ \, \Phi^\text{rel}_H = \{  (\Phi^t_H(l), t, -H(l)) \mid l \in L, t \in \mathbb R \} \colon M \to T^*[-1]\mathbb R.
\end{equation}
This relation lifts the family of Lagrangians $\Phi^t_H(L)\subset M$ to a Lagrangian cobordism.\footnote{Such Lagrangian cobordisms are known as \emph{Lagrangian suspensions}. They appeared first in \cite[Sec~2]{AudinLalondePolterovich1994}. } We can perform the transfer of $\rho$ over $\widetilde{L}_H$ using BV Fubini \eqref{eq:BVFubini}
\begin{equation}\label{eq:TdepHalfdensity}
 \int_{\widetilde{L}_H } \rho = \int_{(L\times \text{id})} \int_{\Phi^\text{rel}_H} \rho =(-1)^{|\rho|}\left[
\tau \int_L (\Phi^t_H)^*(\rho) + \int_L H\cdot (\Phi^t_H)^*(\rho)\right]\sqrt{dt d\tau}.\end{equation}
For closed $\rho$, the RHS is a closed family of half-densities on $T^*[-1]\mathbb R$, i.e.\ a homotopy in the sense of Definition \ref{def:homotopy}. For infinitesimal $H$, we get
\begin{equation}\label{eq:TdepHalfdensityinfinitesimal}
 \int_{\widetilde{L}_H } \rho =(-1)^{|\rho|}\left[
\tau \int_L(\rho + t\Lie_{\{H, -\}} \rho) + \int_L H\cdot \rho\right]\sqrt{dt d\tau}.\end{equation}
The term in the square brackets agrees with the linear interpolation \eqref{eq:interpolationHalfDens} between $\int_L \rho$ and $\int_{L_H}\rho$ coming from the primitive $-\int_L H\cdot \rho$.

\paragraph{Gauge-fixing fermions}
We will now argue that the existence of the relation \eqref{eq:Lagrsuspensionlift} generalizes the fact that turning on the gauge-fixing fermions gives unchanged BV expectation values and equivalent effective BV actions.
\medskip

Indeed, if we flow the Lagrangian $X\subset T^*[-1]X$ by a Hamiltonian $H(q)$, we get a Lagrangian submanifold given by the equations 
\[ \pi_i = \frac{\partial H}{\partial q^i}. \]
Thus, in this case $H$ plays the role of the \emph{gauge-fixing fermion}, see e.g.\ \cite[(15.9.5)]{weinberg}. The version where an (infinitesimal) $H$ depends on $q$ and $\pi$ was considered by \cite[Section~3.3]{hata_zwiebach}.
In these references, it is proven directly that the change of the integral $\int_L \rho$ is proportional to $\Delta \rho$. We can read this off  \eqref{eq:TdepHalfdensity}: 
\begin{align*}
 \Delta_{T^*[-1]\mathbb R} \int_{\widetilde{L}_H} \rho  &\stackrel{\eqref{eq:BVischainmap}}{=}   \int_{\widetilde{L}_H} \Delta_M\rho  , \quad \text{which expands to} \\\partial_t  \int_{\Phi^t_H (L)} \rho &=-
   \tau \underbrace{\int_{\Phi^t_H (L)} \Delta\rho}_{=0 \text{ by \eqref{eq:integral_of_exact_zero}}} - \int_{\Phi^t_H (L)} H\Delta\rho = -\int_{\Phi^t_H (L)} H\Delta\rho.
\end{align*}

Instead of flows of Lagrangians $L \subset M$, we can consider flows of relations $R \colon M \to N$. Composing with $\Phi^\text{rel}_H$, we get a $T^*[-1]\mathbb R$-dependent Lagrangian relation as in \eqref{eq:time-dependentrelation}
\[\begin{tikzcd}
	M & {M\times T^*[-1] \mathbb R} & {N\times T^*[-1] \mathbb R}
	\arrow["{\Phi^\text{rel}_H}", from=1-1, to=1-2]
	\arrow["{R \times \text{id}}", from=1-2, to=1-3]
\end{tikzcd}.\]
Transferring along this composition, we get a family of half-densities on $N$ parametrized by $T^*[-1]\mathbb R$. In physical terms, this means that e.g.\ the effective actions 
\[ e^{W_t/\hbar} \sqrt{dx_p}:= {\int_{\Phi^t_H (L)}} e^{S/\hbar} \sqrt{dx} \]
are related by \emph{canonical transformations} in the sense of \eqref{eq:physicsycanonicaltransformation}. See e.g.\ \cite[Lemma~2.3.1]{costello:renormalisation_and_bv} and \cite[Eq.~(20)]{cattaneo_mnev:BVpushforward}, or \cite[Prop.~2]{cattaneo_mnev:remarks_on_cs} for a normalized version of this statement.

Two striking examples of such flows of Lagrangians submanifolds appeared recently in \cite{lysov:chern-gauss-bonnet_theorem} and \cite{arvanitakis:localization}. 
In \cite{lysov:chern-gauss-bonnet_theorem}, the integral of the Euler class over a compact Riemannian manifold is expressed as a BV integral over a Lagrangian $L_0$. This Lagrangian is then extended to a family using a flow of a Hamiltonian $\psi^i \partial_i W(x)$ \cite[(Eq.~8.2)]{lysov:chern-gauss-bonnet_theorem}. In the limit $t\to \infty$, the BV integral localizes to the Euler characteristic.

Regarding \cite{arvanitakis:localization}, let us describe the simpler, off-shell version from this paper: Also there, a symplectic flow given by a Hamiltonian $-\Psi$ \cite[Eq.~3.42]{arvanitakis:localization} is used to deform a Lagrangian submanifold into a family $L_t$ depending on $t\in [0, \infty)$; the $t\to \infty$ limit integral localizes. The integral is computed by applying this flow to the quantum master action \cite[Eq.~3.43]{arvanitakis:localization}, as we did in \eqref{eq:TdepHalfdensity}.
%\paragraph{Higher categorical structure}
%{\textcolor{gray}Arguably, it is more natural to see half-densities parametrized by $T^*[-1]\mathbb R$ as $2$-morphisms between one-morphisms. The calculation above becomes
%}% https://q.uiver.app/#q=WzAsNSxbMCwwLCJcXGJ1bGxldCJdLFsxLDAsIk0iXSxbMiwwLCJNIl0sWzMsMCwiTiJdLFsyLDEsIlxcYnVsbGV0Il0sWzAsMSwiXFxyaG8iXSxbMSwyLCJcXHRleHR7Z3JhcGh9X3tcXFBoaV9IfSIsMix7ImN1cnZlIjoyfV0sWzEsMiwiXFx0ZXh0e2lkfSIsMCx7ImN1cnZlIjotMn1dLFsyLDMsIlIiXSxbNyw2LCJcXFBoaV5cXHRleHR7cmVsfV9IIiwwLHsib2Zmc2V0IjoyLCJzaG9ydGVuIjp7InNvdXJjZSI6MjAsInRhcmdldCI6MjB9fV1d
%\[\begin{tikzcd}[sep=large]%
%	\bullet & M & M & N \\
%	\arrow["\rho", from=1-1, to=1-2]
%	\arrow[""{name=0, anchor=center, inner sep=0}, "{\text{graph}_{\Phi_H}}"', curve={height=12pt}, from=1-2, to=1-3]
%	\arrow[""{name=1, anchor=center, inner sep=0}, "{\text{id}}", curve={height=-12pt}, from=1-2, to=1-3]
%	\arrow["R", from=1-3, to=1-4]
%	\arrow["{\Phi^\text{rel}_H}", shift right=2, between={0.2}{0.8}, Rightarrow, from=1, to=0]
%\end{tikzcd}\]

\section{Batalin-Vilkovisky fiber integrals over linear relations and homological perturbation theory}\label{sec:linearBV}
We will now specialize from the general BV formalism above to case of \emph{perturbative BV fiber integrals over linear Lagrangian relations between vector spaces\footnote{We would like to note that the restriction to linear Lagrangian relations is too severe to include interesting examples of symplectic flow relations as in Section \ref{sec:Hamflows}.
}}.

As the previous one, also this section is a review of published results, with the exception of the homotopies in Section \ref{sec:homotopiesfromHPL}.

\begin{definition}

We say a $\mathbb Z$-graded real vector space $V$ is \emph{$(-1)$-shifted symplectic} if it comes with an antisymmetric non-degenerate pairing $\omega\colon V \otimes V \to \mathbb R$ of degree $-1$, i.e.\ it vanishes on arguments of total degree different from $1$. We call such space \emph{differential graded} if it is moreover equipped with an anti-self-adjoint differential $q\colon V \to V$ of degree $1$; in other words
\[ \omega(qv, w) + (-1)^{|v|}\omega(v, qw) = 0. \qedhere\]
\end{definition}
See Appendix \ref{app:lineartodifgeom} on how to relate the conventions for dg $(-1)$-shifted symplectic vector spaces with manifolds from the previous section.

\subsection{Cyclic special deformation retracts and Lagrangian relations}

As before, $\overline{V}$ denotes $V$ with $-\omega$ as its symplectic form. 
Each linear Lagrangian relation $R \subset \overline{V}\times U$ which is surjective can be canonically decomposed as (see e.g.\ \cite[Prop.~2.22]{jpz:lagr_rel_and_quantum_Linfty})
\begin{equation} \label{eq:factorizingsurjectiverelation} V \xrightarrow{\text{red}_C} C/C^\omega \xrightarrow{\cong} U \end{equation}
where $C = \text{Im}{R^T} \subset V$, the left image of $R$, is a coisotropic subspace, i.e.\ $C^\omega \subset C$. The relation $\text{red}_C$, the coisotropic reduction, is explicitly given by
\[ \text{red}_C = \{ (c, c \text{ mod } C^\omega) \mid c \in C\} \subset \overline{V}\times C/C^\omega. \]
The symplectic isomorphism $C/C^\omega \cong U$ is determined by $R$: namely by sending $u\in U$ to $(v \text{ mod } \ker R)$ such that $(v, u)\in R$.
\medskip

\begin{definition}
If $V$ is dg $(-1)$-shifted symplectic, we call a Lagrangian relation $R\colon V \to U$ (and also the isotropic subspace $\ker R \subset V$) \emph{non-degenerate} if the symmetric, degree 0 pairing 
\begin{equation}\label{eq:Sfree}
\Sfree(v, w) \equiv (-1)^{|v|}\omega{(qv, w)}
\end{equation} is non-degenerate on $\ker R$ (see Section \ref{sec:convergence}). Given a non-degenerate surjective Lagrangian relation $R \colon V \to U$, one can define a formal BV fiber integral from half-densities on $V$ to half-densities on $U$ as in \cite[Sec.~3.4]{jpz:lagr_rel_and_quantum_Linfty}.
\end{definition}

\medskip

Non-degenerate surjective Lagrangian relations can be more equivalently, if more verbosely, encoded as \emph{cyclic special deformation retracts\footnote{Called \emph{symplectic} in \cite{jpz:lagr_rel_and_quantum_Linfty}, but after discussions with experts we decided to use \emph{cyclic}.}}, which are standard in homological approaches to BV formalism \cite[Definition~5.12]{kajiura2003} \cite{chuang_lazarev:hodge_decomposition} \cite{cattaneo_mnev:remarks_on_cs}.
\begin{definition} \label{def:SDR}
A special deformation retract (SDR) between dg vector spaces from $V$ to $U$ is given by chain maps $p \colon V \rightleftarrows U\colon i$ and a \emph{homotopy} $k\colon V \to V$ of degree $-1$ such that
\begin{equation}
\begin{array}{wc{6em}wc{6em}wc{6em}}
    1_V - ip = kq + qk,  & & pi = 1_U, \\
    k^2 = 0, & pk =0, & ki =0.
\end{array}
\end{equation}
If $V$ and $U$ are both dg $(-1)$-shifted symplectic (this means that $q$ satisfies $q^\dagger = -q$), we call such SDR \emph{cyclic} if 
\begin{equation}
\begin{split}
    &k \text{ is self-adjoint}, \\
    &i^\dagger = p.\qedhere 
\end{split}
\end{equation}
A special deformation retract is called \emph{strict} if $U$ is a symplectic vector subspace of $V$, $i$ is the inclusion and $p$ is the symplectic orthogonal projection. Any special deformation retract is uniquely isomorphic to a strict one $V \rightleftarrows \Im ip$ via $p \colon \Im ip \xrightarrow{\cong} U$.
\end{definition}
Note that $k$ is self-adjoint since it is a partial inverse to an odd anti-self-adjoint operator $q$. Then $ip$ and $kq +qk$ are orthogonal projectors onto complementary symplectic subspaces of $V$ and the images of $kq$ and $qk$ are complementary Lagrangians in the image of $kq+qk$.\footnote{It is important that $k$ and $q$ have opposite behaviour under ${}^\dagger$, as then $(kq)^\dagger= - q^\dagger k^\dagger = q k$, where we have to include a Koszul sign. This is a general statement about symplectic and Lagrangian subspaces: projectors onto symplectic subspaces are self-adjoint, while projectors onto a pair of complementary Lagrangians get exchanged under taking adjoints.} Let us fix the following notation for future reference
\begin{equation}\label{eq:dec}
    V = \underbrace{\Im ip}_{V'} \oplus \underbrace{\Im kq \oplus \Im qk}_{V''},
\end{equation}
the two symplectic subspaces $V'=i(U)$ and $V''$ then induce a decomposition of the BV operator as $\Delta = \Delta' + \Delta''$, and similarly the Poisson bracket.
\medskip

The promised equivalence of SDRs and Lagrangian relations is:
\begin{prop}[{\cite[Ch.~5.~Sec.~2.7]{costelloBook}, \cite[Prop.~3.15]{jpz:lagr_rel_and_quantum_Linfty}}]
There is a bijection between non-degenerate surjective Lagrangian relations $V\to U$ and cyclic SDR's from $V$ to $U$. The space $U$ is isomorphic to cohomology $H(V, q)$ if and only if $qkq = q$, in which case there is a canonical isomorphism given by the composition $i(U)\hookrightarrow \ker q \to H(V, q)$. \end{prop}
\begin{proof}[Sketch of proof]
This bijection is given by sending an SDR into the relation
\begin{equation}\label{eq:SDRtoRELATION} R(k)=\{(v-qkv, p(v)) \mid v\in V\}, \end{equation}
This means the non-degenerate isotrope $\ker R(k)$ is $\operatorname{Im} k$ and the corresponding coisotrope is $\ker{k}$.

In the other direction, given a relation $R$, we use that $\ker R \oplus q(\ker R)\subset V$ is a symplectic subspace written as a sum of two complementary Lagrangians, with $q$ being an isomorphism $\ker{R} \to q(\ker R)$; the homotopy operator $k$ is given by
\begin{equation}\label{eq:k}
    k \text{ is the inverse to } q \colon \ker{R} \xrightarrow{\cong} q(\ker R).
\end{equation} The maps $p$ and $i$ are the projection and inclusion to the symplectic complement of $\ker R \oplus q(\ker R)$, composed with the isomorphism $(\ker R \oplus q(\ker R))^\omega \cong C/C^\omega \cong U$, c.f.\ \eqref{eq:factorizingsurjectiverelation}.
\end{proof}
\begin{remark}
    In \cite{jpz:lagr_rel_and_quantum_Linfty}, we spoke about non-degenerate isotropic subspaces. These are in bijection with \emph{strict} cyclic SDR. A general cyclic SDR is, in addition, equipped with the datum of a symplectic chain isomorphism between $\operatorname{Im} ip$ and $U$; the same datum is needed to upgrade a non-degenerate isotropic subspace of $V$ to a non-degenerate surjective Lagrangian relation $V\to U$.
\end{remark}

%A perturbative version of the BV fiber integral can be then defined as in \cite[Definition 3.22.]{jpz:lagr_rel_and_quantum_Linfty}, satisfying the chain map \eqref{eq:BVischainmap} and Fubini  \eqref{eq:BVFubini} properties.

\subsection{Homological perturbation theory}\label{sec:hpl}
Special deformation retracts can be used to transfer deformations of differentials via the homological perturbation lemma \cite[Th.~1]{Shih1962} \cite{BrownLemma}. There is a natural perturbation of the differential $q$, given by $\hbar \Delta$, if we extend the special deformation retract from $V$ and $U$ to formal polynomial functions\footnote{Note that for $V$ a vector space, the space of functions carries a canonical BV operator, induced by choosing any translation-invariant half-density $\rho_0$. We denote this operator by $\Delta$ instead of $\Delta_0$. It would be interesting to obtain the BV fiber integral on half-densities directly from HPL, following \cite{severa:origin_of_bv} and \cite{CCGN}.} of $V$ and $U$, defined as $\F V := \widehat{\operatorname{Sym}}V^* [[\hbar]]$.\footnote{One could be more precise, see \cite[Def.~3.6]{jpz:lagr_rel_and_quantum_Linfty} or \cite[Sec.~2.2]{doubek_jurco_pulmann:quantum_L_infty_and_HPL}.}
This is the standard \emph{tensor trick} \cite{EM, doubek_jurco_pulmann:quantum_L_infty_and_HPL}, we denote the extended maps by capital letters. 
\vspace{-2mm}
\begin{equation}\label{diag:sdr_on_functions}
\begin{tikzcd}
      \arrow["K"', loop, distance=3em, in=215, out=145, shift right=1]  \left( \F V , Q \right) \arrow[r, "P", shift left] & ( \F U , Q_U ) \arrow[l, "I", shift left]
    \end{tikzcd}
\end{equation}
\vspace{-2mm}

\noindent where $I = p^*$, $P = i^*$ are extended to symmetric tensor powers as algebra maps. The transpose of an odd map involves a sign; on $\phi \in V^*$ we define $Q$ by\footnote{This extension of $q$ to $V^*$ makes the evaluation map $V \otimes V^* \to \mathbb{R}$ into a chain map.}
\begin{equation} \label{eq:Qminustranspose}
    Q\phi  = -q^t (\phi) = - (-1)^{\deg{\phi}} \phi \after q
\end{equation}
and extend it to polynomial functions by the Leibniz rule; it agrees with $\{ \Sfree , - \}$, c.f.\ \eqref{eq:Sfree}. Let $K_\mathrm{un}$ denote the extension of\footnote{Note that $a^t b^t = (-1)^{\deg{a}\deg{b}}(ba)^t$ and similarly $(a^t)^{-1} = (-1)^{\deg{a}}(a^{-1})^t$. This means that \[(QK+KQ)(\phi)=\phi \after (qk+kq) = \phi \after (1_V - ip)(\phi) = (1_{\Fun V}-IP) (\phi),\]
which is a special case of one of the axioms of a SDR.} $k^t (\phi) = (-1)^{\deg{\phi}}\phi \after k,$ again by the Leibniz rule; the homotopy $K$ of \eqref{diag:sdr_on_functions} is defined as
\begin{equation} \label{eq:Knorm} K = \frac{1}{\#''} K_\mathrm{un} \end{equation}
where $\#''$ counts the number of coordinates on $V''$ in a monomial, c.f.\ \eqref{eq:dec}.
Homological perturbation lemma gives explicit formulas for a \emph{perturbed} SDR, which simplify to ($I$ is indeed unchanged)
\vspace{-2mm}
 \begin{equation}\label{diag:sdr_perturbed}
\begin{tikzcd}
      \arrow["K'"', loop, distance=3em, in=215, out=145, shift right=6]  ( \F V , Q + \hbar \BV  ) \arrow[r, "P'", shift left] & ( \F U , Q_U+\hbar \BV_U  ). \arrow[l, "I", shift left]
    \end{tikzcd}
\end{equation}
\vspace{-2mm}

Let us relate $P'$ to  the BV fiber integral\footnote{We follow the arguments from \cite[Sec.~4.1.1]{doubek_jurco_pulmann:quantum_L_infty_and_HPL}. See \cite[Sec.~2.3]{costello:renormalisation_and_bv} \cite{gwilliam-thesis} and the references in \cite{doubek_jurco_pulmann:quantum_L_infty_and_HPL, jpz:lagr_rel_and_quantum_Linfty} and also the recent \cite{cattaneo_mnev:BVpushforward}.}. The HPL formula for the perturbed projection can be simplified as  
\begin{equation}\label{eq:perturbed_P}
    P' \stackrel{\text{HPL}}{=} P (1 - \hbar \Delta K + (\hbar \Delta K)^2 -\dots) = P (1 - [\hbar \Delta, K] + ([\hbar \Delta, K])^2 -\dots),
\end{equation}
where we use the side conditions $K^2=0$ and $PK=0$ to introduce the commutators. If $\alpha$, $\beta$ and $\gamma$ are coordinates on the images of $ip$, $qk$ and $kq$, we have schematically $\Delta = \partial_\beta \partial_\gamma$ and $K = \tfrac{1}{\#_{\beta+\gamma}}\beta \partial_\gamma$. We get $[\Delta, K] = \partial_\gamma \partial_\gamma$, and moreover $P(\beta)=P(\gamma)=0$. Thus, the $k$th term of $P'$ is thus nonzero only on monomials with no $\beta$s and $2k$ of $\gamma$'s (and any number of $\alpha$'s); the normalization of $K$ gives a factorial term and we get 
\begin{equation}\label{eq:perturbed_P_commutator}
    P' = Pe^{ -\tfrac 12 {[\hbar \Delta, K_\mathrm{un}]} }.
\end{equation}
Thus, $P'$ should be seen as the following normalized version of the BV fiber integral along $R(k)$ (denoted as in \eqref{eq:SDRtoRELATION})
\[P'(F) = \frac{\int_{R(k)} F e^{\Sfree/\hbar}}{\int_{R(k)} e^{\Sfree/\hbar}}.\]
The factor $e^{\Sfree/\hbar}$ explains why this map intertwines differentials $\{\Sfree, -\} + \hbar \Delta$, since 
\[ (\{\Sfree, F\} + \hbar \Delta F ) e^{\Sfree/\hbar} = \hbar\Delta(F e^{\Sfree/\hbar}).
\]
The effective BV action (inducing a quantum and a cyclic $L_\infty$ algebra on $U$) is thus given by
\[ W = \hbar \log P'(e^{\Sint/\hbar}) \]
and satisfies the quantum master equation in the form $(\{\Sfr{U}, -\} + \hbar \Delta_U) e^{W/\hbar}  =0 $.

To extend the integral to formal half-densities, one has to choose normalization of the integral of linear half-densities with Gaussian weight. One possible choice is \cite[Def.~3.16]{jpz:lagr_rel_and_quantum_Linfty}.

\subsubsection{Homotopies of BV actions from the homological perturbation lemma} \label{sec:homotopiesfromHPL}
Let us now describe two cases when the homological perturbation lemma immediately gives homotopies between formal half-densities.

\paragraph{Homotopy on V.} Let us choose two different cyclic special deformation retracts with the same source $(V, \omega, q)$; denote one of them with a tilde.   From the homological perturbation lemma, we get a homotopy between $\Sfree + I(W)$ and $\Sfree + \tilde{I}(\tilde{W})$ on $V$. Since \eqref{diag:sdr_perturbed} is a special deformation retract, we have
\[ 1 - I P' = [K', Q+\hbar\Delta] \quad \text{ and } \quad 1 - \widetilde{I} \widetilde{P}' = [\widetilde{K}', Q+\hbar\Delta]. \]
Applying on $e^{\Sint/\hbar}$ and taking the difference
        \begin{equation}\label{eq:homotopy_W_tilde_W}
            e^{I(W)/\hbar} - e^{\widetilde{I}(\tilde{W})/\hbar}  =(Q+\hbar\BV )(-K' + \widetilde{K}')(e^{\Sint/\hbar}),
        \end{equation}      
        so $F=(\widetilde{K}'-K')(e^{\Sint/\hbar})$ is the desired primitive.

        \smallskip
        
        A homotopy between half-densities is obtained by further twisting by $e^{\Sfree/\hbar}$, i.e.\   
        \[ e^{(\Sfree+I(W))/\hbar}\sqrt{dV} - e^{(\Sfree+\widetilde{I}(\widetilde{W}))/\hbar}\sqrt{dV} = \hbar \Delta \left[ e^{\Sfree/\hbar}(-K' + \widetilde{K}')(e^{\Sint/\hbar}) \sqrt{dV}\right]. \]
        A special case is when $W$ is the effective action on cohomology, and $\widetilde{W}=\Sint$, i.e.\ the tilde'd SDR is trivial (the identity relation between $V$ and $V$). Since homotopic actions are related by a symplectic flow, we get a decomposition theorem implemented by a (formal, nonlinear) symplectic diffeomorphism of $V$ \cite[Sec~4.3]{doubek_jurco_pulmann:quantum_L_infty_and_HPL}.

            \paragraph{Homotopy on the retract $U$.} We would also like to compare actions on the \say{smaller} space $U$. We can do this if two strict SDRs satisfy $U =  \Im ip = \Im \tilde{i}\tilde{p}$. This is equivalent to $[q, k-\widetilde{k}]=0$.

        \smallskip
        
        Using the chain map property for $P'$, post-composing equation \eqref{eq:homotopy_W_tilde_W} with $P'$ gives
                \begin{equation}
                 P' I (e^{W/\hbar}) - P'\tilde{I}(e^{\tilde{W}/\hbar})  =
                 (Q_{U}+\hbar\BV_{U })P'(-K' + \widetilde{K}')(e^{\Sint/\hbar})
            \end{equation}
            where $P'K'=0$. Now $\Im ip = \Im \tilde{i}\tilde{p}$, i.e.\ by the retract property $P' I  = 1_U$ of the SDR \eqref{diag:sdr_perturbed},
            \begin{equation}
                 e^{W/\hbar} - e^{\tilde{W}/\hbar}  =
                 (Q_{U}+\hbar\BV_{U})P'\tilde{K}'(e^{\Sint/\hbar}).
            \end{equation}
 In Corollary \ref{cor:homologyequiv} below, we prove that such homotopies exist for any two effective actions on cohomology.

\section{Renormalization group flow as a homotopy}\label{sec:ergf}

Now we will show that the renormalization group flow  in BV formalism of \cite{costelloBook} can be understood as a homotopy of half-densities. In particular, this gives a new\footnote{This can also be done using the homological perturbation lemma, as in Section \ref{sec:homotopiesfromHPL}.} way to produce a homotopy between a quantum master action and its effective action.

\smallskip

The connection between homotopies of quantum master actions and the renormalization group equation was suggested in \cite[Sec~3.2, Sec~4]{zucchini:bv}. We construct this homotopy explicitly; we fix the \say{failure} of exactness of the RG flow referred to in \cite{zucchini:bv} as a \say{seed term} by using a certain twist in Section \ref{sec:twist}. This twist switches between a varying BV Laplacian and a homotopic family of actions. See Section~\ref{sec:zucchini} for more details on this and the connection to the homotopies of BD algebras in \cite{costello_gwilliam:fact_II}.

\smallskip

A finite dimensional version of the following arguments appeared in a talk by one of the authors \cite{jurco:talk2021}.

\subsection{Scale-dependent BV formalism of Costello}\label{costello}

We will use the exact renormalization group flow techniques in the heat kernel formulation, using the notation and results of Costello \cite[Chap.~6,10]{costello:renormalisation_and_bv}, \cite[Chap.~5.8,5.9]{costelloBook}. This formalism is developed to handle infinite dimensional theories and ultraviolet divergences, we will however suppress the subtleties of functional analysis.\footnote{This essentially boils down to the existence of the heat kernels, which is ensured by the properties of $\qGF$. Note that this relies on the choice of Euclidean signature.} We will specialize to the case of a finite dimensional theory in Section \ref{sec:ergf_findim}. In this section, we will consider $V$ to be the (infinite dimensional) vector space of sections of a $\Z$-graded vector bundle over a smooth manifold.\footnote{One should either consider a compact manifold or work with rapidly decaying functions. We will suppress these analytic details.} We fix a free BV theory, i.e.\ a skew self-adjoint operator $q : V \to V$ of degree $1$ which squares to zero and makes $V$ into an elliptic complex, see \cite[Ch.~5.~Def.~7.0.1]{costelloBook}.
\smallskip

Let us recall from \cite{costelloBook} scale dependent effective quantities, namely a BV operator $\BV_L$, an effective interaction term of $S_L$ and a propagator which intertwines $e^{S_L/\hbar}$ at different scales $L$. This scale should be interpreted as the minimal allowed length of propagators between interaction vertices in the worldline formulation of quantum field theory. At $L=0$, one imagines a high energy microscopic theory (after renormalization), which is problematic to formulate in BV formalism. At large $L$, we are describing the low energy effective limit; the renormalization group flow amounts to \emph{zooming out} by integrating out small-scale interactions.

 \smallskip

To do this, we need to fix a gauge fixing operator\footnote{In \cite{costelloBook}, this is denoted $\QGF$. We reserve the capital letter for its extension to $\Fun V$.} $\qGF : V \to V$ of degree $-1$, which we require to be a self-adjoint operator with respect to $\omega$ of order $\leq 1$ which squares to zero, such that 
    \[
    D \define [q,\qGF]
    \]
    is a generalized Laplacian operator in the sense of \cite{costelloBook}, \cite{BGV}. Given these assumptions, for $L \in [0,\infty)$, there exists $\K_L$, the heat kernel for the operator $e^{-L D}$, i.e.\ the quadratic element $\K_L \in V \otimes V$ such that
   \[
   \K_L * = e^{-L D} : V \to V,
   \]
   where $*$ is a contraction using the shifted symplectic pairing. We use the following sign convention, for $P = \sum P' \otimes P'' \in V\otimes V$, we define $P * v \define (-1)^{\deg{v}} \sum P' \omega( P'' , v ).$ Note that with this convention,
 \begin{align}
     ((q\otimes1+1\otimes q)P)* & = [q,P*], \label{eq:Q_convolution} \\
     ((\qGF\otimes1-1\otimes \qGF)P)* & = [\qGF,P*]. \label{eq:QGF_convolution}
 \end{align}
 The difference in signs is caused by $q$ being skew self-adjoint and $\qGF$ being self-adjoint.

   \smallskip
   For $0 < \varepsilon \leq L < \infty $ we define the propagator from scale $\epsilon$ to scale $L$ as the quadratic element
   \begin{equation}
     P_{\varepsilon,L} \define \int_\epsilon^L (\qGF \otimes 1) \K_l \d l \in V \otimes V.
   \end{equation}
    A $\varepsilon > 0$ is necessary in infinite dimension, where at scale $\varepsilon=0$ both the propagator $\partial_{P_{0,L}}$ and the BV Laplacian are ill-defined because of ultraviolet divergences. This is discussed extensively in \cite{costelloBook}.
     
    \smallskip
    The operators $q$ and $\qGF$ are first extended to $V^*$ as $q$ and $k$ in Section \ref{sec:hpl}, and then to functions as derivatives (unlike the normalized $K$ in \eqref{eq:Knorm}). We will denote the contraction with $v \in V$ with $\phi\in V^*$ by $\partial_v \phi \define (-1)^{\deg{v}} \phi(v)$ and we extend it by the Leibniz rule to $\F{V}$. We denote the second order differential operator defined by a graded symmetric quadratic element $P = \sum P' \otimes P'' \in \Sym^2 V$ by $\partial_P = \frac12 \partial_{P'} \partial_{P''} : \F{V} \to \F{V}$. With these sign conventions,\footnote{These sign conventions differ from \cite{costelloBook}, which produces different signs below, e.g.\ in Lemma \ref{lemma:propagator_chain}.} for any graded symmetric quadratic tensor $P \in \Sym^2 (V)$,
   \begin{align}
    [Q, \partial_P ] & = \partial_{(q \otimes 1 + 1 \otimes q)P}, \label{eq:Q_partialP} \\
    [\QGF , \partial_P ] & = - \partial_{(\qGF \otimes 1 + 1 \otimes \qGF)P}. \label{eq:QGF_partialP}
   \end{align}
   We define the scale $L$ BV Laplacian as
   \begin{equation}
       \BV_L \define -\partial_{\K_L}.
   \end{equation}
   We arrive at the following key result, which essentially affirms that the BV formalism is compatible with the renormalization group equation, see also \cite[Ch.~5.~Lemma~9.2.2]{costelloBook}.
   \begin{lemma}\label{lemma:propagator_chain}
   The operator $e^{-\hbar\partial_{P_{\varepsilon, L}}}$ intertwines twisted BV operators at scales $\varepsilon$ and $L$,
    \begin{equation}\label{eq:propagator_chain}
     (Q+\hbar\BV_L) e^{-\hbar\partial_{P_{\varepsilon,L}}}  =
        e^{-\hbar\partial_{P_{\varepsilon,L}}} (Q+\hbar\BV_\varepsilon). 
    \end{equation}
   \end{lemma}
   \begin{proof}
       Follows from $[Q,\partial_{P_{\epsilon,L}}]= \Delta_L - \Delta_\varepsilon$, which is proven using equations \eqref{eq:Q_convolution}, \eqref{eq:Q_partialP}.
   \end{proof}

    \paragraph{The effective interaction.} Following Costello's philosophy, a quantum field theory should be understood as a free theory $\Sfree$ together with a family of effective interactions $S_L$, i.e.\ asymptotically local\footnote{In the sense of \cite[Ch.~5.~Def.~13.4.1]{costelloBook}.} action functionals on $V$ satisfying the \emph{renormalization group equation}
   \begin{equation}\label{eq:ergf}
   e^{S_L / \hbar} =
   e^{-\hbar\partial_{P_{\varepsilon,L}}} e^{S_\varepsilon / \hbar},
   \end{equation}
   and quantum master equation at scale $L$,
   \begin{equation}\label{eq:qme_L}
   (Q+\hbar\BV_L) e^{S_L / \hbar} = 0
   \end{equation}
   for all $ L \in (\varepsilon, \infty)$.
   Note that thanks to equation \eqref{eq:propagator_chain}, if $S_L$ satisfies the quantum master equation at any $L$, then it is satisfied for all $L$. In practice, given an action $S=\Sfree+\Sint$ with ultraviolet divergences, we would define the effective interaction as
      \begin{equation}
   e^{S_L / \hbar} = \lim_{\epsilon\to 0}
   e^{-\hbar\partial_{P_{\varepsilon,L}}} e^{\Sint / \hbar - S^{\mathrm{CT}}_\varepsilon /\hbar},
   \end{equation}
   where the counter-terms $S^{\mathrm{CT}}_\varepsilon$ are constructed as usually, inductively in powers of $\hbar$ and polynomial powers.  See \cite[Chap.~2.~Sec.~13]{costelloBook}. 
   
   \smallskip
   The heat kernel $\K_\infty$ and therefore also the effective action at infinite scale $S_\infty$, which restricts to the effective action $W$ on cohomology, exists only if one requires additional \emph{positivity conditions} of $D$ \cite[Ch.~5.~Def.~10.7.1]{costelloBook}.

   \subsection{Twisting}\label{sec:twist}
  In order to obtain homotopic family of actions, we now introduce a twist of the effective interaction part at scale $L$
   \[
   e^{\tilde{S}_L  / \hbar} \define e^{-\frac12 L D}e^{S_L/\hbar}
   \]
   and similarly a twist of the BV Laplacian at scale $L$, as in \cite[Eq.~4.3.8]{zucchini:bv},
   \[
   \tilde{\BV}_L  \define e^{-\frac12 L D}\BV_L e^{\frac12 L D}.
   \]
    Since $Q$ commutes with $e^{-\frac12 L D}$, $\tilde{S}_L$ satisfies a twisted quantum master equation,
    \begin{equation}\label{eq:twisted_qme}
    (Q+\hbar \tilde{\BV}_L )e^{\tilde{S}_L/\hbar}=0.    
    \end{equation} 

    We will need the following calculation.
    \begin{lemma}\label{lemma:propagator_anticom}   The two second order differential operators $\Delta_L, \partial_{P_{\varepsilon,L}}$ satisfy

       \begin{equation}\label{eq:propagator_anticom}
       \frac{\d}{\d L} \partial_{P_{\epsilon,L}} =  \frac12 [\BV_L,\QGF ].
   \end{equation}
   \end{lemma}
   \begin{proof}
       
   By definition, $\frac{\d}{\d L} \partial_{P_{\epsilon,L}} = \partial_{(\qGF\otimes1)\K_L}$, and using equation \eqref{eq:QGF_partialP},
   \[
   [\BV_L,\QGF ]=   [\QGF ,\BV_L]= -[\QGF , \partial_{\K_L}] = \partial_{(\qGF \otimes1 + 1\otimes \qGF )\K_L} = \partial_{2(\qGF \otimes1)\K_L}.
   \]
   In the last step, we used that $(\qGF \otimes1 - 1 \otimes \qGF )\K_L* = [\qGF ,e^{-L D}]=0,$ which follows from equation \eqref{eq:QGF_convolution}.
   \end{proof}
\begin{prop}
    The twisted BV Laplacian does not depend on the scale $L$, i.e.\ 
    \[ \frac{\d}{\d L} \tilde{\BV}_L = 0.\]
\end{prop}
\noindent
We will thus drop the index $L$ and use $\tilde{\BV}$ below.
\begin{proof}
      First, we will need to compute $\frac{\d}{\d L} \BV_L$, for which we will use the following observation,
       \begin{equation}
           \frac{\d}{\d L} \K_L + (q\otimes1+1\otimes q)(\qGF \otimes1 )\K_L = 0.
       \end{equation}
       This means
       \[
       -\frac{\d}{\d L}\BV_L + \partial_{(q\otimes1+1\otimes q)(\qGF \otimes1 )\K_L} = 0
       \]
       and by equation \eqref{eq:Q_partialP} and rewriting $(\qGF \otimes1 )\K_L= \frac{\d}{\d L} P_{\varepsilon,L}$, we have
       \[
       \frac{\d}{\d L}\BV_L = \left[ Q, \frac{\d}{\d L} \partial_{P_{\varepsilon,L}} \right].
       \]
       By equation \eqref{eq:propagator_anticom}, we have
       \begin{align*}
           \frac{\d}{\d L}\BV_L & = \frac12 \left[ Q ,  [\BV_L , \QGF ] \right] \\
           & = \frac12 [Q , \Delta_L \QGF  + \QGF  \Delta_L]   \\
           & = \frac12 (D \BV_L - \BV_L D).
       \end{align*}
       The last equality holds by $[Q,\Delta_L]=-\partial_{(q\otimes1+1\otimes q)\K_L}=0$, which follows from $(q\otimes1+1\otimes q)\K_L*=[q,e^{-L D}]=0$. Finally, we get
       \[
       \frac{\d}{\d L} \tilde{\BV} =(-\frac12 e^{-\frac12 L D} D\BV_L e^{\frac12 L D} + e^{-\frac12 L D}\frac12 [D,\BV_L]e^{\frac12 L D} + \frac12 e^{-\frac12 L D} \BV_L D e^{\frac12 L D})=0.
       \]
\end{proof}

    \subsubsection{Renormalization group flow as homotopy for twisted interaction}
This leads us to the main result of this section.
   \begin{prop} The change of $e^{\tilde{S}_L/\hbar}$ is exact, namely
    \begin{equation}\label{eq:homotopy_ergf_twisted_action}
    \frac{\d}{\d L} e^{\tilde{S}_L/\hbar} =
    -\frac12 (Q+\hbar\tilde{\BV}  ) {Q}^{GF}e^{\tilde{S}_L/\hbar}.    
    \end{equation}
   \end{prop}
\begin{proof}
  First, by the renormalization group equation \eqref{eq:ergf},
   \begin{equation}\label{eq:ergf_inf}
       \frac{\d}{\d L} e^{S_L/\hbar} = - \hbar \frac{\d}{\d L} ( \partial_{P_{\varepsilon
,L}} ) e^{S_L/\hbar}.
   \end{equation}
   Using equation \eqref{eq:propagator_anticom} and the quantum master equation at scale $L$ \eqref{eq:qme_L} we get
    \begin{align*}
       \frac{\d}{\d L} e^{S_L/\hbar} & = -\frac12 (\QGF \hbar\BV_L+\hbar\BV_L \QGF ) e^{S_L/\hbar} \\
       & = -\frac12(-\QGF Q+\hbar\BV_L \QGF ) e^{S_L/\hbar}  \\
       & = -\frac12(-D +Q\QGF +\hbar\BV_L \QGF ) e^{S_L/\hbar}  \\
       & = \frac12(D-(Q+\hbar\BV_L ) \QGF )e^{S_L/\hbar}.
   \end{align*}
 Replacing $e^{S_L/\hbar}$ with the twist $e^{\tilde{S}_L/\hbar}$ and using the fact that $D$ commutes with both $Q$ and $\QGF $, we have
    \begin{equation*}
       \frac{\d}{\d L} e^{\tilde{S}_L/\hbar} =
       -\frac12 e^{-\frac12 L D} D e^{S_L/\hbar} + e^{-\frac12 L D}\frac12(D-(Q+\hbar\BV_L ) \QGF )e^{S_L/\hbar} 
       = -\frac12 (Q+\hbar \tilde{\BV}  ) {Q}^{GF}e^{\tilde{S}_L/\hbar}.
   \end{equation*}
\end{proof}

    \subsubsection{Finite dimension: homotopy to the effective action}
\label{sec:ergf_findim}

    Let us now specialize to the finite dimensional setting and $\qGF$ coming from special deformation retracts to put the results of this section in the context of homological perturbation theory from Section \ref{sec:hpl}.

    \smallskip
    
    A free BV theory is given by a finite-dimensional dg $(-1)$-symplectic vector space $(V,q)$ as in Section \ref{sec:linearBV}. The operator $\qGF$ can be chosen as the homotopy operator $k$ \eqref{eq:k} of some SDR, or equivalently it comes from a choice of a non-degenerate isotrope $I = \Im kq$. The conditions on existence of heat kernels are automatically satisfied; the operator $D = [q,\qGF ]$ reduces to the projector onto $\Im kq \oplus \Im qk$. On $\F{V}$, this operator counts the number of basis elements of the dual of $\Im kq \oplus \Im qk$ in a given monomial, when expressed in a basis respecting the decomposition. The scale-dependent BV operator $\Delta_L$ is $\Delta' + e^{-L} \Delta''$; the twisted operator $\tilde{\Delta}$ is equal to $\Delta$ at $L=0$ and thus at all $L$.

    \smallskip
   
     Given a quantum master action with a free part $Q$ and an interaction part $\Sint \in \F{V}$, we can define the effective interaction part at scale $L$ as
     \begin{equation}
        e^{S_L/\hbar} \define e^{-\hbar\partial_{P_{0,L}}}e^{\Sint/\hbar}.
   \end{equation}
   The actions $\tilde{S}_L$ satisfy the scale $0$ quantum master equation, and their flow \eqref{eq:homotopy_ergf_twisted_action} is an \say{honest} homotopy with respect to $Q+\hbar\Delta$.

    \smallskip
    
    For $\varepsilon=0, L=\infty$, equation \eqref{eq:propagator_chain} recovers equation \eqref{eq:BVischainmap} for the BV fiber integral. It can be shown that, indeed, the propagator gives the normalized BV fiber integral;
   \begin{equation}\label{eq:propagator_int}
    P e^{-\hbar\partial_{P_{0,\infty}}} = \frac{\int_{R(k)} e^{\Sfree/\hbar} (-) \dhalf V}{\int_{R(k)} e^{ \Sfree/\hbar} \dhalf V}.   
   \end{equation}
   To prove this, we use the following lemma.
\begin{lemma}\label{lemma:propagator_commutator}
With $\qGF  = k$, we can express the propagator as a graded commutator,
       \[\partial_{P_{0,\infty}} = \frac12 [\BV, \QGF ].\] 
\end{lemma}
\begin{proof}
   Using equations \eqref{eq:Q_partialP} and \eqref{eq:Q_convolution},
    \begin{equation}
           [\BV, \QGF ] = -[ \partial_{\K_0}, \QGF ]
           = \partial_{2 (\qGF \otimes 1 )  \K_0}.
    \end{equation}
    \noindent
    As in \cite[Lemma~6.5.1]{costello:renormalisation_and_bv}, 
    $P_{0,\infty}* = k = \qGF$ and using that $\K_0$ is the kernel of the identity,
       \[
       [\BV, \QGF ] = 2 \partial_{P_{0,\infty} * \K_0} = 2 \partial_{P_{0,\infty}}.\qedhere
       \]
\end{proof}

The comparison with the normalized BV integral (via HPL) is a simple corollary.
\begin{prop}
    With $\qGF  = k$, the perturbed projection \eqref{eq:perturbed_P} can be expressed using Costello's propagator as
    \[
P' = P e^{-\hbar \partial_{P_{0,\infty}}}.
\]
\end{prop}
\begin{proof}
    By equation \eqref{eq:perturbed_P_commutator}, $P' = P e^{-\frac12 [\hbar\Delta, K_\mathrm{un}]}$, and Lemma \ref{lemma:propagator_commutator} finishes the proof.
\end{proof}

Finally, we get an alternative homotopy to \eqref{eq:homotopy_W_tilde_W}.
\begin{prop}
    The actions $e^{(\Sfree+\Sint)/\hbar}$ and $e^{(\Sfree + I(W))/\hbar}$ are homotopic, i.e.\ they differ by a $\Delta$-exact term.
\end{prop}
\begin{proof}
Integrating the results of equation \eqref{eq:homotopy_ergf_twisted_action} from $0$ to $\infty$, we get a homotopy
   \begin{equation}\label{eq:homotopy_from_ergf}
    e^{\tilde{S}_\infty/\hbar} - e^{\tilde{S}_0/\hbar} = (Q+\hbar\BV)\left( -\frac12 \int_0^\infty \QGF e^{\tilde{S}_l/\hbar} \dif l \right)    
   \end{equation}
   between $\tilde{S}_0 = S_0 = \Sint$ and $\tilde{S}_\infty = W$. The last identification holds by equation \eqref{eq:propagator_int},
   \[
   \lim_{L\to\infty}e^{-\frac12 L D}e^{S_L/\hbar} = \proj_R e^{S_{\infty}/\hbar} = \proj_R e^{-\hbar\partial_{P_{0,\infty}}}e^{\Sint/\hbar} = e^{I(W)/\hbar}.
   \]
   To get a homotopy of the half-densities in the sense of Definition \ref{def:homotopy}, i.e.\ find a primitive of $\Delta$, we can twist equation \eqref{eq:homotopy_from_ergf} by $e^{\Sfree/\hbar}\sqrt{\d V}$ like in Section \ref{sec:homotopiesfromHPL}.
   \qedhere
\end{proof}

   \subsubsection{Relation to the work of Costello, Gwilliam, and Zucchini}\label{sec:zucchini}
   It was noted in the proof of Lemma 11.1.1 in \cite[Ch.~5]{costelloBook} that the BV Laplacian
   \[
   \BV_L + \d L \partial_{\frac{\d}{\d L}P_{\epsilon,L}} =  -\partial_{\K_L} + \d L \partial_{\qGF \K_L} 
   \]
   
    on $\F{V} \otimes \Omega^\bullet ( (0,\infty) )$ encodes both the quantum master equation and the renormalization group flow. In particular, $S_L \in \F{V} \otimes C^\infty ( (0,\infty) )$ solves the quantum master equation with respect to this BV Laplacian if and only if both the scale $L$ quantum master equation \eqref{eq:qme_L} and the infinitesimal version of the renormalization group equation \eqref{eq:ergf_inf} hold. This BV Laplacian is used to define a $\Omega(\Delta^k)$-linear family of BD algebras in \cite[Sec.~8.8.2]{costello_gwilliam:fact_II} in the parametrix formalism. This is interpreted by Costello and Gwilliam as a homotopy of BD algebras induced by the renormalization group flow. Here, however, the effective action $I[\phi]$ (defined on top of page 123 of \cite{costello_gwilliam:fact_II}) has no higher form\footnote{In other words, the term $\varphi_t$ from Definition \ref{def:homotopy} is missing.} components. We present a version of this construction with a fixed BV Laplacian with no higher form components which is a homotopy of the scale-dependent effective actions in the sense of Definition \ref{def:homotopy}. In other words, we twisted away the scale dependence of the BV Laplacian and absorbed it in $\tilde{S}_L$.
    
    In \cite[Sec.~4]{zucchini:bv}, it is shown that if $S_L + S^\star_L \d L \in \F{V} \otimes \Omega^\bullet ( (0,\infty) )$, where $S^\star_L$ is a degree $-1$ element, satisfies the quantum master equation with respect to this BV Laplacian, the renormalization group flow of the scale $L$ effective interaction is $(Q+\hbar\BV_L)$-exact up to a certain \say{seed term}. It is discussed in \cite[Sec.~3.2.]{zucchini:bv} that up to this failure, this could be interpreted as a homotopy. In particular, from \cite[Eqs.~4.67, 4.84]{zucchini:bv} it can be seen that this failure is proportional to the action of the operator $D$ (denoted $\mathcal{H}$ in \cite[Sec.~4]{zucchini:bv}) on the scale $L$ effective interaction, as in the proof of equation \eqref{eq:homotopy_ergf_twisted_action}. We solved this issue by using the $e^{-\frac12 LD}$ twist \cite[Eqs.~4.3.5,4.3.8]{zucchini:bv}, and obtained a $(Q+\hbar\BV_0)$-exact flow of the twisted interaction $\tilde{S}_L$. We did not assume the existence of a degree $-1$ partner, in fact, it is constructed in equation \eqref{eq:homotopy_ergf_twisted_action}.

\section{The space of cyclic special deformation retracts} \label{section:parametrizing-gauge-fixings.}

Let us now characterize the space of possible gauge fixings in the linear finite-dimensional BV formalism. As we have seen in Section \ref{sec:linearBV}, such gauge fixings are given by cyclic special deformation retracts, or equivalently nondegenerate surjective Lagrangian relations. To give such relation $V\to U$, we have to specify two pieces of data:
\begin{itemize}
    \item A nondegenerate isotropic subspace $J\subset V$, or equivalently a \emph{strict} cyclic special deformation retract on $V$ (Def.~\ref{def:SDR}).
    \item A dg symplectic isomorphism $J^\omega/J \xrightarrow{\cong} U$.
\end{itemize}
In this section, we will first give a convenient parametrization of the space of cyclic SDRs, strict and nonstrict, using the \emph{Hodge decomposition} of $V$ \eqref{eq:dec}. With this parametrization, we will first characterize infinitesimal deformations of SDRs, recovering the result of Cattaneo and Mn\"ev \cite{cattaneo_mnev:remarks_on_cs}. Then we look at strict SDRs onto cohomology and prove that this space is contractible. This will allows us to conclude that all effective action on cohomology, computed using a strict SDR, are homotopy equivalent. We also explain in Section \ref{sec:shapeofSDR} that this is, in some sense, the best one can hope for.

Without the assumption of cyclicity, the contractibility of SDRs onto cohomology is easier to prove, and was established by Mn\"ev and Cattaneo--Mn\"ev--Reshetikhin in \cite[Lemma~1]{mnev:discreteBF}, \cite[Lemma~4.4~(a)]{cattaneo_mnev_reshetikhin:cellular_tft}. Such SDRs are given by a choice of complements to  $\Ker q \subset V$ and $\Im q \subset \Ker q$, and these choices are naturally parametrized by affine spaces of sections. In the present (cyclic) case, the problem becomes non-linear because the complements need to be isotropic and symplectic, respectively. Our proof does not use complements, but instead a parametrization of the homotopy operator $k$.

\subsection{Parametrizing cyclic special deformation retracts}
    Let us begin by giving a minimal set of data and axioms needed to characterize special deformation retracts.
    \begin{prop}\label{prop:minimalpresentationSDR}
			A strict cyclic SDR on $V$ is equivalently given by degree $-1$ self-adjoint linear map  $k$ such that
			$k^2 = 0$ and $kqk = k$. 
            \medskip
            
            A non-strict cyclic SDR $V\to U$ is equivalently given by map $k$ as above and a chain map $p\colon V \to U$ such that\footnote{The relation $p p^\dagger = 1_R$ is equivalent to saying that $p$ is a Poisson map, or that $i = p^\dagger$ is symplectic.} $p p^\dagger = 1_U$ and $p^\dagger p = 1 - [q, k]$.
		\end{prop}
        \begin{proof}  
            In the first case, it is easy to see that $J = \Im k$ is a non-degenerate isotrope, and the induced SDR has $k$ as its homotopy operator. The decomposition $J \oplus qJ \oplus (J \oplus qJ )^\omega$ agrees with $\Im k q \oplus \Im q k\oplus \Im t$, where $t = 1-[q, k]$.
			 
           In the non-strict case, $p$ is a dg symplectic isomorphism $(\Im t, q_r)$ to $(U, q_U)$ and $i = p^\dagger$ is the inverse.
        \end{proof}

	Let $k_0$ be a fixed homotopy operator of a strict special deformation retract on a dg vector space $(V, q)$. This SDR will serve as a reference against which we compare other SDRs. It induces a decomposition $V = \Im k_0 q \oplus \Im qk_0 \oplus \Im t_0$ where $t_0 = i_0p_0$ \eqref{eq:dec}. In this decomposition, the linear maps of the SDR can be written as 
	\begin{equation*}
		k_0 = \begin{pmatrix}
			0 & k_0 & 0 \\
			0& 0& 0 \\
			0 & 0& 0
		\end{pmatrix}, \quad
		q = \begin{pmatrix}
			0 & 0& 0\\
			q & 0 & 0 \\
			0 & 0 & q_r
		\end{pmatrix}, \quad		
		t_0 =  \begin{pmatrix}
			0 & 0& 0\\
			0 & 0 & 0 \\
			0 & 0 & 1
		\end{pmatrix}.		
	\end{equation*}
	where $k_0\colon \Im qk_0 \to \Im k_0q$ is an inverse to the isomorphism $q \colon \Im k_0q \to \Im qk_0$. The map $q_r = t_0 q t_0 \colon \Im t_0 \to \Im t_0$ is the differential on the reduced space $\Im t_0$; it is zero if and only if $qk_0q = q$, in which case the reduction is canonically isomorphic to the cohomology $H(V, q)$.  
    
    To get a non-strict cyclic SDR $V\to U$, we furthermore need a map $p_0 = (0, 0, p_0)^\text{T}$ where the component $p\colon \Im t_0 \to U$ is a symplectic dg isomorphism.
    \bigskip

    We now take a different SDR given by $k$ (and $p$ in the nonstrict case). Let us start with the condition that $k$ is self-adjoint. In general, an operator $O \colon V \to V$ can be written as  \[O = (k_0q + qk_0 + t_0) O (k_0q + qk_0 + t_0) = \text{nine terms}.\]
     These get shuffled around by taking adjoints as follows:
		% https://q.uiver.app/#q=WzAsMTUsWzEsMSwiXFxidWxsZXQiXSxbMiwyLCJcXGJ1bGxldCJdLFsyLDEsIlxcYnVsbGV0Il0sWzEsMiwiXFxidWxsZXQiXSxbMywzLCJcXGJ1bGxldCJdLFsyLDMsIlxcYnVsbGV0Il0sWzMsMSwiXFxidWxsZXQiXSxbMywyLCJcXGJ1bGxldCJdLFsxLDMsIlxcYnVsbGV0Il0sWzEsMCwiXFxJbSBrZCJdLFsyLDAsIlxcSW0gZGsiXSxbMywwLCJcXEltIHQiXSxbMCwxLCJcXEltIGtkIl0sWzAsMiwiXFxJbSBkayJdLFswLDMsIlxcSW0gdCJdLFswLDEsIiIsMCx7InN0eWxlIjp7InRhaWwiOnsibmFtZSI6ImFycm93aGVhZCJ9fX1dLFsyLDIsIiIsMCx7InJhZGl1cyI6MX1dLFszLDMsIiIsMCx7InJhZGl1cyI6MSwiYW5nbGUiOi05MH1dLFs0LDQsIiIsMCx7InJhZGl1cyI6MSwiYW5nbGUiOjEzNX1dLFs1LDYsIiIsMCx7InN0eWxlIjp7InRhaWwiOnsibmFtZSI6ImFycm93aGVhZCJ9fX1dLFs3LDgsIiIsMCx7InN0eWxlIjp7InRhaWwiOnsibmFtZSI6ImFycm93aGVhZCJ9fX1dXQ==
		
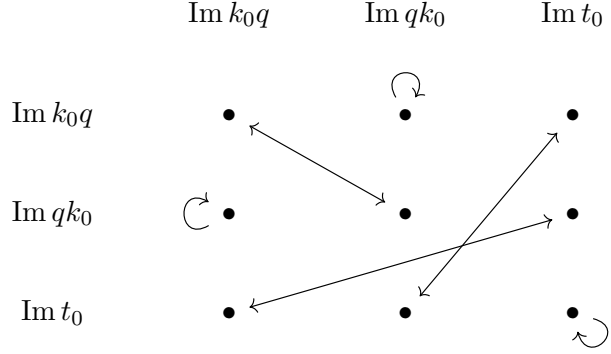
\[\begin{tikzcd}
			& {\Im k_0q} & {\Im qk_0} & {\Im t_0} \\
			{\Im k_0q} & \bullet & \bullet & \bullet \\
			{\Im qk_0} & \bullet & \bullet & \bullet \\
			{\Im t_0} & \bullet & \bullet & \bullet
			\arrow[tail reversed, from=2-2, to=3-3]
			\arrow[from=2-3, to=2-3, loop, in=60, out=120, distance=5mm]
			\arrow[from=3-2, to=3-2, loop, in=150, out=210, distance=5mm]
			\arrow[tail reversed, from=3-4, to=4-2]
			\arrow[tail reversed, from=4-3, to=2-4]
			\arrow[from=4-4, to=4-4, loop, in=285, out=345, distance=5mm]
		\end{tikzcd}\]
		
	For example, $k_0qOt_0$ is sent to $(k_0qOt_0)^\dagger = t_0 O^\dagger qk_0$. Therefore, we can parametrize a self-adjoint $k$ as follows\footnote{This $D$ is not related to $D$ from Section \ref{sec:ergf}}
	\[ k = \begin{pmatrix}
		A & B & C \\
		D & A^\dagger & E \\
		E^\dagger & C^\dagger & F 
	\end{pmatrix}, \quad \text{where } B^\dagger = B, D^\dagger = D \text{ and } F^\dagger = F.\]
	Note that, when reading $B^\dagger = B$, we understand $B$ as a map $V \to V$ when computing its adjoint.

    Similarly, for non-strict SDRs, let us parametrize $p$ by
    \[ p = \begin{pmatrix}
        P_1 & P_2 & P_3 
    \end{pmatrix} , \quad \text{i.e.} \quad i = \begin{pmatrix} P_2^\dagger \\ P_1^\dagger \\ P_3^\dagger \end{pmatrix}. \] 
    
	Let us now record the non-linear conditions on $k$ and $p$ from Proposition \ref{prop:minimalpresentationSDR}. For strict SDRs, we have $kqk = k$ giving
	\begin{equation}
	\begin{pmatrix}
	BqA + Cq_rE^\dagger & BqB + Cq_rC^\dagger & BqC + Cq_r F \\
	A^\dagger q A + E q_r E^\dagger & A^\dagger qB + E q_r C^\dagger & A^\dagger q C + E q_r F \\
	C^\dagger q A + Fq_rE^\dagger & C^\dagger qB + F q_r C^\dagger & C^\dagger q C + F q_r F
	\end{pmatrix} \stackrel{!}{=} \begin{pmatrix}
		A & B & C \\
		D & A^\dagger & E \\
		E^\dagger & C^\dagger & F 
	\end{pmatrix}
	\end{equation}
	and $k^2 =0$ giving
	\begin{equation}
	\begin{pmatrix}
		A^2 + BD+ CE^\dagger & AB + BA^\dagger+CC^\dagger & AC +BE+CF\\
		DA + A^\dagger D + E E^\dagger & DB + (A^\dagger)^2 + EC^\dagger & DC + A^\dagger E + EF \\
		E^\dagger A + C^\dagger D + FE^\dagger & E^\dagger B + C^\dagger A^\dagger + FC^\dagger & E^\dagger C + C^\dagger E + F^2
	\end{pmatrix} \stackrel{!}{=}0.
	\end{equation}
    For non-strict SDRs, we have in addition
    \begin{equation}
        P_1 P_2^\dagger + P_2 P_1^\dagger + P_3 P_3^\dagger = 1_U
    \end{equation}
    and 
    \begin{equation}
        \begin{pmatrix}
            P_2^\dagger P_1 & P_2^\dagger P_2 & P_2^\dagger P_3 \\
            P_1^\dagger P_1 & P_1^\dagger P_2 & P_1^\dagger P_3 \\
            P_3^\dagger P_1 & P_3^\dagger P_2 & P_3^\dagger P_3 
        \end{pmatrix}
        = \begin{pmatrix}
            1-Bq & 0 & -Cq_r \\
            -A^\dagger q -qA & 1-qB & -Eq_r - qC \\
            -C^\dagger q -q_r E^\dagger & -q_rC^\dagger & 1-Fq_r - q_r F
        \end{pmatrix}.
    \end{equation}
    These equations are too complicated to be practical. We will now study two special cases: when $k$ is close to $k_0$, and when both $k$ and $k_0$ are SDRs onto cohomology.

    \subsection{Infinitesimal deformations}\label{sec:infdef}
    Assuming that $k$ is a small deformation of $k_0$, we get that all parameters $A \dots P_3$ are first order in a deformation parameter, \emph{apart from} $B = k_0 + \delta_B$ and $P_3 = p_0 + \delta_P$ with $\delta_B$ and $\delta_P$ small. Keeping only the first order, we get the following equations:
    \begin{itemize}
        \item $kqk = k$ implies
	\[  
	\begin{pmatrix}
	A & 2\delta_B & C  \\
	0 & A^\dagger & 0 \\
	0  & C^\dagger  & 0
	\end{pmatrix} = \begin{pmatrix}
	A & \delta_B & C  \\
	D & A^\dagger & E \\
	E^\dagger  & C^\dagger  & F
	\end{pmatrix}.
	\]
    i.e.\ $\delta_B$, $D$, $E$ and $F$ are automatically zero.

    \item $k^2 =0$ gives one additional condition 
    \[ Ak_0 + k_0 A^\dagger = 0. \]
    \item $pp^\dagger=1$ gives
    \[ p_0 \delta_P^\dagger + \delta_P p_0^\dagger = 0. \] 
    \item $[q, k] = 1 - p^\dagger p$ implies
    \[ P_1 = -p_0C^\dagger q, \quad P_2 = -p_0 q_r C^\dagger .\]
    \item The chain map condition of $p$ is automatically satisfied for $P_1$, $P_2$ and implies that $\delta_P$ is also a chain map.
    \end{itemize}
    \begin{prop}\label{prop:infinitesimal} Infinitesimal deformations of an SDR $V \to U$ are given by a triple of maps
    \[ A \colon \Im k_0 q \to \Im k_0 q , \quad C\colon \Im t_0 \to \Im k_0 q, \]
    and a chain map
    \[ \quad \delta_P \colon (\Im t_0, q_r) \to (U, q_U) \]
    such that
    \[ Ak_0 + k_0 A^\dagger = 0, \quad  p_0 \delta_P^\dagger + \delta_P p_0^\dagger = 0\]
    and $|A| = |C| = -1$, $|\delta_P| = 0$.
    \end{prop}
    In other words, $Ak_0 \colon \Im qk_0 \to \Im k_0 q$ is self-adjoint and $\delta_P p_0^\dagger \colon U \to U$ is anti-self-adjoint.

    \subsubsection{Dimension of the space of cyclic special deformation retracts}
    The space of strict deformation retracts is decomposed into components according to the graded dimension of the non-degenerate isotropic subspace $J:=\Im k$. Since $\Sfree$ is a non-degenerate pairing on $J$, then $\dim J_i = \dim J_{-i}$; denote $x_i = \dim J_i$. The maximal possible dimension of $J_i$, denoted $d_i$, is such that the reduction along $J$ is onto cohomology; i.e.\ $d_i + d_{i-1} + b_i =  \dim V_i$,  where $b_i = \dim H_i(V, q)$ are the Betti numbers. It was proven in the thesis of Michal Vorobel \cite[Tvrdenie~13]{Vorobel} that these are sufficient conditions on the existence of $J$: for any sequence $x_i$ such that 
    \[ x_i = x_{-i} \quad \text{and} \quad x_i \le d_i,\]
    there is a non-degenerate isotrope $J$ with $\dim J_i = x_i$. This is proven by diagonalising the pairing $\Sfree$. Loc.\ cit.\ \cite[Veta~17]{Vorobel} also computes the dimension of the space of non-degenerate isotropes as
    \[ \sum_i x_i(\dim V_i - x_i - x_{i-1}/2). \]
    We can recover this formula from the number of free parameters in $C$ and $A$ as in Proposition \ref{prop:infinitesimal}. The dimension of the space of cyclic special deformation retracts has an additional term coming from $\delta_P$, which depends on the differential $q_U$. When $q_U=0$, i.e. we reduce to cohomology, the additional number of parameters is  $\sum_{\mathbb Z} b_i^2/2$.

    \subsubsection{Three types of deformations of Cattaneo and Mn\"ev}
    Let us now explain that these three parameters correspond to the three types of deformations of Cattaneo and Mn\"ev \cite[Statement~4]{mnev:discreteBF}, \cite[Eq.~(4)]{cattaneo_mnev:remarks_on_cs}. In these sources, the non-degeneracy is not assumed, but since it is an open condition,\footnote{I.e.\ it picks out an open subset.} it does not give additional conditions on small deformations we consider in this section.\footnote{ For convenience, we indicate how to translate from our notation to that of \cite{cattaneo_mnev:remarks_on_cs}
    \begin{align*}
    V \;&\leadsto\; \mathcal F \\
    i \colon U \hookrightarrow V \;&\leadsto\; \iota \colon \mathcal F'\hookrightarrow F \\
    V'' = i(U)^\omega = \Im k_0 q \oplus \Im q k_0 \subset  V \;&\leadsto\; \mathcal F'' \subset \mathcal F \\
    \Im k_0 q \subset V'' \;&\leadsto\; \mathcal L \subset F''.
    \end{align*}}

    \begin{itemize}
        \item The map $A$ is used to deform the two Lagrangians in $\Im qk_0 \oplus \Im k_0 q$ by defining new non-degenerate isotrope
        \[ [1 + A^\dagger q] (\Im k_0 q) \subset  \Im qk_0 \oplus \Im k_0 q,\]
        while $\Im q k_0 $ is not deformed, and also the symplectic decomposition $V = V' \oplus V''$ does not change. This is Type I of \cite{cattaneo_mnev:remarks_on_cs}. It can be parametrized by a gauge-fixing fermion $\Psi$, a quadratic function on $\Im k_0 q$.
        
        \item The map $C$ changes the embedding of $R$, and also necessarily the decomposition of its complement. This is Type II of \cite{cattaneo_mnev:remarks_on_cs}. Indeed, the inclusion is deformed to
        \[ i = i_0 - ( C q_r i_0, qC i_0, 0 )^\mathrm{T}, \]
        i.e.\ the \emph{perpendicular part} of \cite{cattaneo_mnev:remarks_on_cs} $\delta\iota_\perp=-C q_r i_0 - qC i_0 \colon U \to V''$. We see that $\delta\iota_\perp$ is parametrized by an unconstrained map $C$. The decomposition of $V''$ into $\Im k_0q$ and $\Im qk_0$ also changes as described in loc.\ cit., for example $\Im k_0q$ changes to 
        \[ [1+ C^\dagger q] \Im k_0 q  \subset \Im k_0 q \oplus \Im i_0 p_0 \]
        
        \item The map $1 + \delta_P p_0^\dagger$ acts by a symplectomorphism of $U$, i.e.\ the decomposition $V = \Im k_0q \oplus \Im qk_0 \oplus \Im t_0$ is unchanged, and only the isomorphism $\Im t_0 \to U$ is deformed by $1+\delta_P p_0^\dagger$. This is Type III of \cite{cattaneo_mnev:remarks_on_cs}.
    \end{itemize}

\subsubsection{Infinitesimal changes of SDRs can be induced by Hamiltonian flows}
Let us take an infinitesimal symplectic automorphism $1+X$ of $V$. It can always be seen as an infinitesimal Hamiltonian flow, the Hamiltonian being
\begin{equation} \label{eq:HamfromX} H_X = \frac{1}{2} (-1)^{|x^a|} \omega_{ab} x^a X^b_c x^c,\end{equation}
where $X(e_a) = X_a^b e_b$.

If we are given a small deformation of a SDR as parametrized above, we can always induce it by $X$ in the sense that the change of $k$ is $[X, k_0]$ and the change of $p$ is $-pX$; while $[X, q] = 0$ \cite[proof of Prop.~2]{cattaneo_mnev:remarks_on_cs}. This corresponds to precomposing the Lagrangian relation $R(k_0) \colon V \to U$ with the graph of the symplectomorphism $1-X$.
\smallskip

If we parametrize this change by $A$, $C$ and $\delta_p$, the explicit form of $X$ inducing such change of the SDR is
\begin{equation}\label{eq:XfromSDRchange}
    X=\begin{pmatrix}
        \xi & 0 & -Cq_r \\
        -qA & -\xi^\dagger & -qC \\
        -C^\dagger q & -q_r C^\dagger & -p^\dagger \delta_P
    \end{pmatrix},
\end{equation}
where $\xi\colon \Im k_0q \to \Im k_0 q$ is an arbitrary degree zero map satisfying $q\xi + \xi^\dagger q =0$.

\subsection{Parametrizing retracts onto cohomology}
    A strict deformation retract is onto cohomology if $qkq = q$.     
    A non-strict retract onto cohomology is parametrized by a strict retract plus a map $p \colon V \to H(V, q)$. The peculiarity of this case is that each strict retract defines \emph{canonical} non-strict retract onto cohomology: the map $p_k$ 
	\[  V \xrightarrow{\mathrm{proj}} \underbrace{\Im qk \oplus \Im t}_{\Ker q} \xrightarrow{\mod \Im q} H(V, q)\]
    i.e.
    \[ p_k(v) = [(1-kq)(v)] \]
	is a projection forming a cyclic SDR.

    We will now restrict to such cyclic SDRs $V\to H(V, q)$ coming from a strict SDR. The conditions for $k$ simplify, especially because $q_R$ is zero and we have an additional condition $qkq = q$. 
    \begin{prop}
    	All strict cyclic SDRs between $V$ and a subspace isomorphic to cohomology are parametrized by a pair of linear degree $-1$ maps $A\colon \Im k_0 q \to \Im k_0 q$ and $C \colon \Im t_0 \to \Im k_0 q$ which satisfy a single equation 
		\begin{equation} \label{eq:strictSDRhomology} Ak_0 + k_0 A^\dagger + C C^\dagger = 0. \end{equation}
    \end{prop}
    \begin{proof}
	Let us apply the conditions in the following convenient order:
    
    From $qkq = q$ we get one explicit condition
		\[ q B q = q. \]
		Since $q$ is an isomorphism $\Im k_0q \to \Im qk_0$, this equation implies that $B = q^{-1} = k_0$.
		
		The equation $k q k = k$ gives the matrix equality (using the previous point)
		\begin{equation*}
 			\begin{pmatrix}
			A & B & C \\
			A^\dagger q A & A^\dagger & A^\dagger qC \\
			C^\dagger qA & C^\dagger & C^\dagger q C
			\end{pmatrix}
			= 			\begin{pmatrix}
				A & B & C \\
				D & A^\dagger &E \\
				E^\dagger & C^\dagger & F
			\end{pmatrix}.
		\end{equation*}
		The components $D$, $E$ and $F$ are determined
		\[ D = A^\dagger q A, \quad E = A^\dagger q C, \quad F = C^\dagger q C, \]
		and there are no further conditions. Note that $k^\dagger=k$ is automatically satisfied.
		
		 The equation $k^2=0$ gives the vanishing of the following block matrix
		\begin{equation*}
			\begin{pmatrix}
				A^2 + k_0A^\dagger q A + C C^\dagger q A & Ak_0 + k_0 A^\dagger + C C^\dagger & AC + k_0 A^\dagger q C +CC^\dagger q C \\
				A^\dagger q A^2 + {A^\dagger}^2 q A + A^\dagger qCC^\dagger qA & \text{--} & A^\dagger q AC + {A^\dagger}^2 qC + A^\dagger q CC^\dagger qC \\
				\text{--} & \text{--} & C^\dagger qAC + C^\dagger A^\dagger qC + C^\dagger qCC^\dagger qC 
			\end{pmatrix}
		\end{equation*}
		The entries marked with -- don't give additional constraints since $k^\dagger = k$.
		Actually, the equation
		\begin{equation} Ak_0 + k_0 A^\dagger + C C^\dagger = 0. \end{equation}
		implies all the other equations.	
    \end{proof}	
	This leads us to the main result of this section.
	\begin{theorem}\label{thm:contractible}
		The space of strict cyclic SDRs with $qkq = q$ is contractible to a point.
	\end{theorem}
	\begin{proof} 
	   Relative to one such SDR (with homotopy operator $k_0$), other such SDRS are given by maps $A$ and $C$ satisfying \eqref{eq:strictSDRhomology}. Thus, we can always connect this SDR to the reference one by a family
	\[A_t = t^2 A, \quad C_t = t C\]
	which preserves the validity of \eqref{eq:strictSDRhomology}. 
	\end{proof}
    \begin{remark}
	The other components of $k$ scale as
	\[ B \sim t^0, \quad D \sim t^4, \quad E \sim t^3, \quad F \sim t^2.\]
	The fact that one can consistently choose scaling powers follows from the fact that for the following matrices
	\[ 
	\text{scalings of $k$} = \begin{pmatrix}
		t^2 &t^0 &t^1\\
		t^4 & t^2 & t^3 \\
		t^3 & t^1 & t^2 
	\end{pmatrix}
	\quad \text{and} \quad
	\text{scalings of $q$}=
		\begin{pmatrix}
		0 & 0 & 0\\
		1 & 0 & 0 \\
		0 & 0 & 0 
	\end{pmatrix}
	\]
	the products corresponding to $kqk$ and $qkq$ have homogeneous entries scaling as $k$ and $q$. This fails if we consider nonzero $q_r$; for example the first entry of $kqk$ is $BqA + Cq_RE^\dagger$, with terms scaling as $t^2$ and $t^4$.
    \end{remark}

    Expressing $k$ back from the parametrizing matrices, we obtain the following interesting interpolation between homotopy operators
    \begin{align*}
        k_t &= (k_0 q + t^2 qk_0 + t i_0p_0) k (t^2 k_0 q + q k_0 + t i_0 p_0)
    \\ &= (1-t)^2 k_0 + t^2 k  \\ &+ t(1-t) \left[ kqk_0 + k_0 qk  - t (k k_0 q + q k_0 k) + (t-1)(k_0qkk_0q + qk_0kqk_0) + t(1-t) qk_0kk_0q\right].
    \end{align*}
    Using that the canonical projection onto cohomology is given by $p_{k_t} = p_0 (1-k_tq)$, we get that 
    \[ p_{k_t} = p_0  (1 - t kq). \]

 Let us call a non-degenerate isotrope maximal if the image of the coisotropic reduction is isomorphic to cohomology (this is equivalent to being maximal in the poset of non-degenerate isotropes).
	\begin{cor}\label{cor:homologyequiv}
		Given two maximal non-degenerate isotropes, the effective actions on cohomology are homotopic.
	\end{cor}
	\begin{proof}
		We can connect two such SDRs by a path as in Theorem \ref{thm:contractible}. As explained in \cite[Theorem~3]{mnev:simplicialBF} \cite[Statement~6]{mnev:discreteBF} \cite[Proposition 2]{cattaneo_mnev:remarks_on_cs} or in our Section \ref{sec:infdef}, a small change of the special deformation retract $R_t\colon V \to H(V, q)$ induces a canonical transformation of the effective action. Indeed, such change of the SDR comes from a small symplectomorphism of $V$ \eqref{eq:XfromSDRchange}, which is given by a Hamiltonian $H_t$ \eqref{eq:HamfromX}. A small change of $R_t$ is thus implemented by the relation $\Phi^\mathrm{rel}_{H_t dt}$ \eqref{eq:Lagrsuspensionlift}. Note that, at fixed time, the relation above is the graph of the inverse of the flow of $\{H, -\}$, which agrees with the fact that a small change of a SDR is given by precomposing with the graph of $1-X$. As we calculated in Section \ref{sec:Hamflows} using BV operators, the time derivative of the effective action $e^{W_t/\hbar}$ on $H(V,q)$ is $\Delta \int_{R(k_t)} H(t)e^{S/\hbar}$, see also \cite[Eq.~(20)]{cattaneo_mnev:BVpushforward}. The full homotopy is thus given by the integral over $t\in [0,1]$ of the half-density $\int_{R(k_t)} H(t)e^{S/\hbar}$. 
	\end{proof}
	
\begin{remark}
    It would be interesting to give a simple expression for the homotopy constructed by the previous proof, possibly along the lines of \cite[Statement~6]{mnev:discreteBF}. The direct application of the procedure outlined above leads to unwieldy formulas.
\end{remark}

\subsection{The space of cyclic special deformation retracts}\label{sec:shapeofSDR}
Let us summarize what we know about the shape of the space of cyclic special deformation retracts. If we fix $V$ and $U$, the space of all cyclic SDRs $V\to U$ is a principal bundle over the space of strict SDRs, with the structure group being the dg symplectic group of $U$. We note that this group can have non-trivial topology: for example for $U=T^*[-1]\mathbb R^n$ with $q_U=0$, the group is $GL_n(\mathbb R)$.

As we have seen, for $U\cong H(V, q)$, the base of this bundle is contractible. For bigger $U$ (i.e.\ smaller non-degenerate isotropes $J$), this is not the case. For example, for $V= T^*[-1]\mathbb R^n$, a differential is given by a quadratic form $\Sfree$ on $\mathbb R^n$, and the space of non-degenerate isotropic subspaces of dimension $k$ is the complement of the vanishing locus of a polynomial function on the Grassmannian $\mathrm{Gr}_k(n)$. For $n=2$, the space of non-degenerate 1-dimensional isotropic subspaces is homeomorphic to $S^1$, $\mathbb R$ or $\mathbb R \sqcup \mathbb R$, depending on the signature of the quadratic form $\Sfree$ \cite[Príklad~7]{Vorobel}. 

If $V= T^*[-1]\mathbb R^n$ and $\Sfree$ is positive (or negative) definite, the space of non-degenerate isotropes of fixed dimension is always connected.

\begin{comment}
    In the non-symplectic setting, special deformation retracts onto cohomology are parametrized by a choice of complements to $\Ker q \subset V$ and $\Im q \subset \ker q$. Since such splittings form an affine space, the space of strict special deformation retracts is always contractible. In our case, with a symplectic form on $V$, we require that the complements are isotropic and symplectic respectively \eqref{eq:dec}, which puts non-linear conditions on the splittings.
\end{comment}

\section{Spans} \label{sec:spans}

    Let us now shortly summarize previous developments. If we start with a quantum master action $S$ on a dg $(-1)$-shifted vector space $(V, \omega, q)$, we can compute effective action for any non-degenerate isotropic subspace $J\subset V$. This is compatible with composition of relations, and we get a whole poset of effective field theories, coming from the poset of nondegenerate isotropic subspaces. Furthermore, if we reduce to cohomology $H(V, q)$, i.e.\ as much as possible, the effective actions are homotopic and thus related by non-linear formal symplectomorphisms.

    \smallskip

    Let us now consider a complementary question. Given two quantum master actions $(V, S)$ and $(\tilde{V}, \tilde{S})$ on two different spaces and a (non-linear) isomorphism relating the effective actions on their cohomologies, \emph{does there exist a common ``ancestor'' vector space with a master action $(X, S_X)$ such that one can get $(V, S)$ and $(\tilde{V}, \tilde{S})$ as BV pushforwards of $S_X$?\footnote{This can be read as asking for a span of loop homotopy Lie or quantum $L_\infty$ algebras \cite{zwiebach:closed_string,markl2001loop}, i.e.\ in this section we answer a question we posed in the linear case in \cite[Remark~4.20]{jpz:lagr_rel_and_quantum_Linfty}. An analogous statement about (non-cyclic) $L_\infty$ algebras was proven by \cite{farahani_saemann_wolf:spans_of_Linfty}; we comment on their approach below in Remark \ref{rem:FSW}.}}
        \[\begin{tikzcd}[ampersand replacement=\&]
	\& {(X,S_X)} \\
	{(V,S)} \&\& {(\tilde{V},\tilde{S})}
	\arrow[from=1-2, to=2-1]
	\arrow[from=1-2, to=2-3]
\end{tikzcd}\]
    
    We will construct such $S_X$ using the non-linear symplectomorphism (a field redefinition) which relates two homotopic actions \cite[Ch.~5.~Sec~10.1.]{costelloBook} \cite[Sec.~4.3]{doubek_jurco_pulmann:quantum_L_infty_and_HPL}, which we now recall shortly: If \[e^{S_1/\hbar} - e^{S_0/\hbar} = \hbar \Delta F\] on $(V, \omega, q)$, then there exists a time-dependent Hamiltonian such that its time $1$ flow $\Theta_F$ satisfies
    \begin{equation} \label{eq:ThetaF} (\Theta_F)_*(e^{S_0/\hbar} \sqrt{dV})= e^{S_1/\hbar}\sqrt{dV}.\end{equation} See \cite[Sec.~2.2]{doubek_jurco_pulmann:quantum_L_infty_and_HPL} for a more careful discussion of the convergence of these maps using the weight grading.
    \medskip
    
    Let us therefore start with the following diagram of BV pushforwards.
    % https://q.uiver.app/#q=WzAsNCxbMCwwLCIoViwgZV57KFxcU2ZyZWUgKyBTKS9cXGhiYXJ9XFxzcXJ0e2RWfSkiXSxbMSwwLCIoViwgZV57KFxcdGlsZGV7U31fXFx0ZXh0bm9ybWFse2ZyZWV9ICsgXFx0aWxkZXtTfSkvXFxoYmFyfVxcc3FydHtkXFx0aWxkZXtWfX0pIl0sWzAsMSwiKFUsIGVee1cvXFxoYmFyfVxcc3FydHtkVX0pIl0sWzEsMSwiKFxcdGlsZGV7VX0sIGVee1xcdGlsZGV7V30vXFxoYmFyfVxcc3FydHtkXFx0aWxkZXtVfX0pIl0sWzAsMiwiUiJdLFsxLDMsIlxcdGlsZGV7Un0iXSxbMiwzLCJcXFBzaSIsMV1d
\[\begin{tikzcd}
	{\left(V, e^{S/\hbar}\sqrt{dV}\right)} & {\left(V, e^{\tilde{S}/\hbar}\sqrt{d\tilde{V}}\right)} \\
	{(U, e^{W/\hbar}\sqrt{dU})} & {(\tilde{U}, e^{\tilde{W}/\hbar}\sqrt{d\tilde{U}})}
	\arrow["R"', from=1-1, to=2-1]
	\arrow["{\tilde{R}}", from=1-2, to=2-2]
	\arrow["\Psi"{description}, from=2-1, to=2-2]
\end{tikzcd}\]
   Here, $(V,S)$ and $(\tilde{V},\tilde{S})$ are quantum master actions, and $R\colon V \to U$, $\tilde{R}\colon \tilde{V} \to \tilde{U}$ are two non-degenerate linear Lagrangian relations relating the corresponding half-densities.  The relation $\Psi$ is the graph\footnote{This is where we are sloppy about (formal) graded geometry in this section; we do not specify in what sense $\Psi$ is a subspace of $\flip{U}\times \tilde{U}$.} of a non-linear formal symplectomorphism $U \to \tilde{U}$ again relating the corresponding half-densities.\footnote{The reason we separate the non-linear symplectic diffeomorphism $\Psi$ is that $R$ and $\tilde{R}$ need to be linear Lagrangian relations. If we allowed non-linear Lagrangian relations, or restrict to linear isomorphisms \cite{jpz:lagr_rel_and_quantum_Linfty}, we could simply assume that $R$ and $\tilde{R}$ have the same target and the effective actions agree there.} Note that the BV pushforward along a graph of a symplectic isomorphism $\Psi$ is just the pushforward along $\Psi$. 

\begin{theorem}\label{thm:spans}
Let $S, \tilde{S}, R, \tilde{R}$ and $\Psi$ be as above.
 Then there exists a quantum master action $(X, S_X)$, a constant half-density $\sqrt{dX}$ and a pair of non-degenerate \emph{non-linear} surjective Lagrangian relations $T\colon X\to V$, $\tilde{T}\colon X \to \tilde{V}$ such that 
    \[ \int_T e^{S_X/\hbar} \sqrt{dX} = e^{S/\hbar}\sqrt{dV}, \quad \int_{\tilde{T}} e^{S_X/\hbar}\sqrt{dX} = e^{\tilde{S}/\hbar} \sqrt{d\tilde{V}}\]
    and such that the following diagram of BV pushforwards commutes when applied to the action $e^{S_X/\hbar}$.\footnote{This defines a commutative diagram in a non-linear version of the category $*/\QOSC$ from \cite[Sec.~4,Remark~4.17]{jpz:lagr_rel_and_quantum_Linfty}.}
    % https://q.uiver.app/#q=WzAsNSxbMiwwLCIoWCwgU19YKSJdLFswLDEsIihWLCBTKSJdLFsxLDIsIihVLCBXKSJdLFs0LDEsIihcXHRpbGRle1Z9LCBcXHRpbGRle1N9KSJdLFszLDIsIihcXHRpbGRle1V9LCBcXHRpbGRle1d9KSJdLFswLDEsIlIiLDJdLFsxLDIsIlQiLDJdLFswLDMsIlxcdGlsZGV7Un0iXSxbMyw0LCJcXHRpbGRle1R9Il0sWzIsNCwiXFxQc2kiLDFdXQ==
\begin{equation}\label{diag:spanpentagon}\begin{tikzcd}
	&& {(X, S_X)} && \\
	{(V, S)} &&&& {(\tilde{V}, \tilde{S})} \\
	& {(U, W)} && {(\tilde{U}, \tilde{W})}
	\arrow["T"', from=1-3, to=2-1]
	\arrow["{\tilde{T}}", from=1-3, to=2-5]
	\arrow["R"', from=2-1, to=3-2]
	\arrow["{\tilde{R}}", from=2-5, to=3-4]
	\arrow["\Psi"{description}, from=3-2, to=3-4]
\end{tikzcd}\end{equation}
%     There exists quantum master action $(X,S_X)$ and a span of quantum $\Linfty$ morphisms $X \to V$ and $X \to \tilde{V}$.
%     \[\begin{tikzcd}[ampersand replacement=\&]
% 	\& {(X,S_X)} \\
% 	{(V,S)} \&\& {(\tilde{V},\tilde{S})}
% 	\arrow[from=1-2, to=2-1]
% 	\arrow[from=1-2, to=2-3]
% \end{tikzcd}\]
\end{theorem} 
We will call the top half of this diagram, i.e.\ the pair of relations from $(X, S_X)$ to $V$ and $\tilde{V}$ a \emph{span}.\footnote{Note that it is by construction an \emph{orthogonal span} as in \cite[Def.~2.26]{jpz:lagr_rel_and_quantum_Linfty}.} We leave the question of commutativity of \eqref{diag:spanpentagon} as a diagram of Lagrangian relations for future research.

\begin{proof}
Let us take care of the constant half-density factor $\sqrt{dX}$ at the end of this proof.\smallskip

Recall that the non-degenerate reduction $R$ induces a symplectic vector space isomorphism  $V \cong U \oplus T^*[-1]J$ where $J = \ker{R}$; similarly $\tilde{V} \cong \tilde{U} \oplus T^*[-1]\tilde{J}$. In this decomposition, the differential on $V$ decomposes as $q|_J \colon J \to J^*[-1]$ and $q|_U \colon U \to U$. Then $\Sfree$, the kinetic part of the action $S$, splits similarly, with $\Sfree|_J$ being a non-degenerate quadratic polynomial on $J$.

Define 
\[  X \define U \oplus T^*[-1]J \oplus T^*[-1] \tilde{J} \]
and let 
\[S_X \define W + \Sfree|_J + \tilde{S}_\textnormal{free}|_{\tilde{J}},\] which satisfies the quantum master equation since the three terms are quantum master actions on three disjoint symplectic subspaces of $X$. 

The relation $T\colon X\to V$ is given by the following composition
\[ T = X \xrightarrow{T_\text{red}} U \oplus T^*[-1]J \xrightarrow{\cong} V. \]
Here $T_\text{red}$ is the reduction along the non-degenerate isotrope $\tilde{J}$. The second arrow is a non-linear symplectic isomorphism, given by $\Theta_F$ as in \eqref{eq:ThetaF}, coming from (as explained in Section \ref{sec:homotopiesfromHPL})
\[ e^{S/\hbar} - e^{(\Sfree|_J + p^*(W))/\hbar } = \hbar \Delta F,\]
 composed with the linear map $V \cong U \oplus T^*[-1]J$. These two relate the (exponentiated) actions $W+ \Sfree|_J$ and $S$.
 The relation $\tilde{T}$ is given similarly, but we need to insert $\Psi$ as well
\[\tilde{T} = X \xrightarrow{\tilde{T}_\text{red}} U \oplus T^*[-1]\tilde{J} \xrightarrow{\Psi\times 1} \tilde{U} \oplus T^*[-1]\tilde{J}\xrightarrow{\cong} \tilde{V}.\]
In total, we get the following diagram of effective actions.
\[\begin{tikzcd}[column sep=-1em]
	& {\left(X,S_X = W + \Sfree|_J + \tilde{S}_\textnormal{free}|_{\tilde{J}} \right)} & \\
	{\left(U\oplus T^*[-1]J, W + \Sfree|_J \right)} && {\left(U\oplus T^*[-1]\tilde{J}, W + \tilde{S}_\textnormal{free}|_{\tilde{J}} \right)} \\
	&& {\left(\tilde{U}\oplus T^*[-1]\tilde{J},  \tilde{W} + \tilde{S}_\textnormal{free}|_{\tilde{J}}\right)} \\
	{\left(V, S\right)} && {\left( \tilde{V}, \tilde{S}  \right)}
	\arrow["{T_\text{red}}"{description}, from=1-2, to=2-1]
	\arrow["{\tilde{T}_\textnormal{red}}"{description}, from=1-2, to=2-3]
	\arrow["{\Theta_F}"', from=2-1, to=4-1]
	\arrow["{\Psi \times 1}", from=2-3, to=3-3]
	\arrow["{\Theta_{\tilde{F}}}", from=3-3, to=4-3]
\end{tikzcd}\]
\emph{Let us now also find a correct normalization $\sqrt{dX}$.} First, note that $\sqrt{dU}$ is determined by the choice of $\sqrt{dV}$ and the relation $R$; indeed $\sqrt{dU} = \int_R \exp(\Sfree|_J / \hbar)\sqrt{dV}$ \cite[Def.~3.16]{jpz:lagr_rel_and_quantum_Linfty}. Let us also fix a constant half-density $\mu_J$ on $T^*[-1]J$. Then 
\[  \sqrt{dU} \otimes \mu_J = {N_J} \sqrt{dV} \quad \text{where} \quad N_J = \int_{J\subset T^*[-1]J} e^{\Sfree|_J /\hbar} \mu_J.\]
Then, let 
\[\sqrt{dX} := \frac{1}{N_J N_{\tilde{J}}} \sqrt{dU} \otimes \mu_J \otimes \mu_{\tilde{J}}. \]
When integrating along $T_\text{red}$, the factor $N_{\tilde{J}}$ cancels out and we are left with $\sqrt{dV}$, and similarly for $\tilde{T}_\mathrm{red}$.
\end{proof}

\begin{remark} \label{rem:FSW}
    Let us explain how this relates to the work of \cite{farahani_saemann_wolf:spans_of_Linfty}. They also construct a span of transfers from a pair of $L_\infty$ algebras $V_{1,2}$ with isomorphic minimal models on $H$. Their $X$ is bigger in size, it is the resolution
    \[ V_1 \oplus V_2 \oplus H[-1]. \]
    We note that this does not generalize immediately to the cyclic case, since $H[-1]$ does not have any symplectic form in general. Instead, we expect that their formula should work if we choose a Lagrangian splitting of the cohomology $H$ and then take the odd cotangent of a \cite{farahani_saemann_wolf:spans_of_Linfty}-type resolution. The quantum master action on this resolution needs a correction term, as explained in \cite[Prop.~2.5]{farahani_saemann_wolf:spans_of_Linfty}. It would be interesting to find a Lagrangian relation between such symplectic lift of their span and our span.

    In addition, we note that \cite{farahani_saemann_wolf:spans_of_Linfty} don't allow $L_\infty$ non-linear isomorphisms in the setup to their construction (e.g.\ they assume the minimal models coincide up to a strict $L_\infty$ map). The use of non-linear isomorphisms is crucial to our construction, but from our geometric viewpoint they are quite natural.
\end{remark}

\begin{remark} In \cite[Sec.~4.4]{doubek_jurco_pulmann:quantum_L_infty_and_HPL}, we proved that each BV pushforward $V\to U$ sending $S$ to an effective action $W$ also induces a non-linear Poisson map $V\to U$ which intertwines the $S$ and $W$-twisted BV differentials \eqref{eq:QME}. One can similarly convert the diagram \eqref{diag:spanpentagon} supplied by Theorem \ref{thm:spans} to a diagram of non-linear Poisson maps.\footnote{Note that spans of Poisson maps play a role in classical $0$-shifted Poisson geometry—they form \emph{dual pairs} and \emph{Morita equivalences}, see e.g.\ \cite{weinstein:local_structure_poisson,xu:morita_poisson}.}
\end{remark}

\subsection{Spans and composition of relations}

It is natural to ask whether the span construction admits a reasonable notion of composition. First, one can construct spans by iterating application of Theorem \ref{thm:spans}.
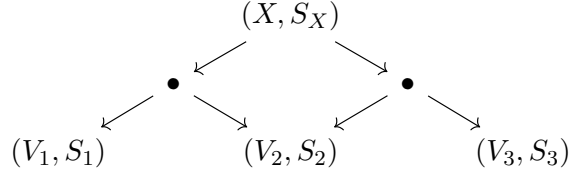
\[\begin{tikzcd}[ampersand replacement=\&,sep=small]
	\&\& (X,S_X) \&\& \\
	\& \bullet \&\& \bullet \\
	{(V_1,S_1)} \&\& {(V_2,S_2)} \&\& {(V_3,S_3)}
	\arrow[from=1-3, to=2-2]
	\arrow[from=1-3, to=2-4]
	\arrow[from=2-2, to=3-1]
	\arrow[from=2-2, to=3-3]
	\arrow[from=2-4, to=3-3]
	\arrow[from=2-4, to=3-5]
\end{tikzcd}\]
Another way is to use composition of relations. In \cite[Sec.~4.4]{jpz:lagr_rel_and_quantum_Linfty}, we considered Lagrangian relations which induce an isomorphism of effective actions, i.e.\ exactly the input to Theorem \ref{thm:spans}. For a pair of such relations $\smash{(V_1, S_1) \xrightarrow{R_1} (V_2, S_2) \xrightarrow{R_2} (V_3,S_3)}$, by \cite[Thm.~4.18]{jpz:lagr_rel_and_quantum_Linfty}, their composition $R_2 \after R_1$ also defines an isomorphism between the effective actions under certain orthogonality and non-degeneracy conditions. We can then construct a span with a master theory $(Y,S_Y)$ over $R_2 \after R_1$.

These two constructions are compatible in the following weak sense: Assuming an appropriate decomposition in the sense of equation \ref{eq:dec}, there is a BV pushforward from the iterated span to the span over $R_2 \after R_1$ making the following diagram commute.
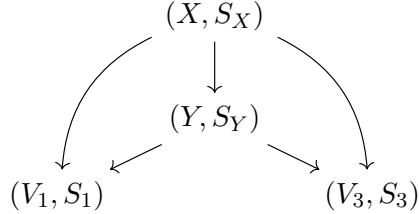
\begin{equation}\label{diag:spans_commute}
    \begin{tikzcd}[ampersand replacement=\&,sep=small]
	\& (X,S_X) \& \\
	\\
	\& (Y,S_Y) \\
	{(V_1, S_1)} \&\& {(V_3,S_3)}
	\arrow[from=1-2, to=3-2]
	\arrow[bend right, from=1-2, to=4-1]
	\arrow[bend left, from=1-2, to=4-3]
	\arrow[from=3-2, to=4-1]
	\arrow[from=3-2, to=4-3]
\end{tikzcd}
\end{equation}
This BV pushforward is given by reduction along the isotropic subspace $\ker R_1^T \cap \ker R_2 \subset X$ composed with the graph of a certain non-linear isomorphism of the form $\Theta_F$, similarly to the proof of Theorem \ref{thm:spans}. For more details, see \cite{zika:thesis}.

\appendix
\section{Graded manifolds and sign conventions}
We are vague about the definition of a graded manifold. The reader is invited to think of $\mathbb Z_2$-graded manifolds with an Euler vector field, supplying the $\mathbb Z$ grading compatible with the parity.

We see forms on a (graded) manifold $X$ as functions on $T[1]X$, i.e.\ $dx$ has a degree $|x|+1$, which goes into the Koszul rule. This rule is broken when talking about degrees of symplectic forms: the degree of $\omega_{ij} dx^i dx^j$ is $|x^i|+|x^j|$ (for $|\omega_{ij}|$ of degree $0$). This means that for $(-1)$-shifted symplectic manifolds, $\omega_{ij}$ is \emph{symmetric} in its indices, since $dx^i$ and $dx^j$ have opposite degree. 
\medskip

Let us now specify the conventions we use for $\mathbb Z$-graded manifolds. We will completely follow the conventions presented in Appendix C.1 of \cite{Kupka2026}. Let us recall from there that Lie derivative is defined as $\Lie_X = [d, i_X]$, $L_X f = Xf$ for $f$ a function. Moreover, this implies $i_X df = (-1)^X X f$. Hamiltonian vector fields are given by $i_{X_H} \omega = - dH$, with $df = dx^i \partial_i f$, and $\{f ,g\} = X_f g$.  We recommend a reader interested in signs to consult \cite{Kupka2026} carefully.
\medskip

Let us now add some additional formulas which follow from these conventions. If we write  $\omega = \tfrac 12 \omega_{ij} dx^i dx^j$ for our $(-1)$-shifted symplectic form, then the Hamiltonian vector field is
\begin{equation}\label{eq:Hamiltonian}
    X_f = (-1)^{|f||x^i|} \omega^{ij} \frac{\partial f}{\partial x^i}\partial_j
\end{equation}
where $\omega^{ij}\omega_{jk} = \delta^i_k$. The Poisson bracket is
\[ \{f, g\} = (-1)^{|f||x^i|} \frac{\partial f}{\partial x^i} \omega^{ij} \frac{\partial g}{\partial x^j}.\]

For constant $\omega_{ij}$, the BV operator on half-densities (and on functions, using $\rho_0=\sqrt{dx_i}$ as the reference half-density), is by \eqref{eq:Deltasecondorder}
\[ \Delta = \frac {1}{2} \omega^{ij} \frac{\partial^2 }{\partial x^i \partial x^j} . \]
Let us also justify the signs in the formula \eqref{eq:CartanMagicFormula} i.e.\ $\Lie_{\{H, -\}} = (-1)^{|H|}[\Delta, H\cdot ]$. There is some ambiguity in determining the sign of a Lie derivative of a (half-)density along an odd vector field. One way to fix a sign is by formally writing, for an arbitrary vector field\footnote{Odd vector fields which don't square to zero don't have flows, see e.g.\ \cite[Footnote~17]{Cattaneo2023BVsemidensities}.} $V$
\begin{equation} \label{eq:pullbackandLie}
(\Phi_V^t)^* = 1 + t \Lie_V + \dots\end{equation}
Applying such pullback to a half-density $\sqrt{dx}$, we get from \eqref{eq:berezinianpullback}
\begin{align*} (\Phi_V^t)^*\sqrt{dx} &= \Ber{ \frac{\partial (\Phi_V^t)^* x}{\partial x} }^{\tfrac 12} \sqrt{dx} \\ &= \Ber{\delta^i_j + (-1)^{|V||x^j|}t\partial_j V^i}^{\tfrac 12}\sqrt{dx}
 \\& = \sqrt{dx} + (-1)^{|V||x^j|+ |x^j|} \tfrac 12 t \partial_iV^i\sqrt{dx} \end{align*}
which gives, again by \eqref{eq:pullbackandLie},
that the Lie derivative of the coordinate half-density is the divergence\footnote{This also agrees with the signs in \cite[Lemma~A.1]{Cattaneo2023BVsemidensities}.}
\begin{equation} \label{eq:divergence} \Lie_V \sqrt{dx}=  (-1)^{(|V|+ 1)|x^i|}\frac 12 \partial_iV^i \sqrt{dx}.\end{equation}
Alternatively, one can check that this is the only sign factor depending on $|V|$ and $|x^i|$ which implies\footnote{For example, one can take \[\Lie_{\tau \partial_\tau} \tau f(t) |dtd\tau|, \quad \Lie_{t\partial_t}\tau f(t) |dt d\tau|, \quad \Lie_{\tau t \partial_t} f(t)|dt d\tau|, \text{ and } \Lie_{\tau_1\tau_2 \partial_{\tau_1}}\tau_1 g(t_1, t_2) |dt_1dt_2d\tau_1d\tau_2|,\] with $f(t)$ and $g(t_1, t_2)$ rapidly decreasing functions.} $\int \Lie_V \rho = 0$.

\smallskip

The Cartan magic formula for half-densities \eqref{eq:CartanMagicFormula} follows by combining \eqref{eq:Hamiltonian} and \eqref{eq:divergence}.

\subsection{Symplectic vector spaces} \label{app:lineartodifgeom}
The conventions for $(-1)$-shifted symplectic vector spaces require the symplectic form to be graded-antisymmetric. Let us explain how to see such a vector space as a graded manifold. Denote $(V, \omega^\text{lin}, q^\text{lin})$ the vector space, and let us choose a basis $e_i$; the dual basis $\phi^i$ generates polynomial functions on the graded manifold $V$. This manifold has the following symplectic form $\omega^\text{geom}$ and $Q^\text{geom}$
\begin{equation}\label{eq:omega_geom_lin}
    \omega^\text{geom} = \frac 12 (-1)^{|x^i|+1}\omega^\text{lin}(e_i, e_j) d\phi^i d\phi^j
\end{equation}
and homological vector field 
\begin{equation}
    Q^\text{geom} = (-1)^{|x^i|+1} Q^i_a \phi^a \partial_i
\end{equation}
where $Q(e_j) = Q^i_j e_i$. One can check that $q^\text{lin}$ being anti-self-adjoint is equivalent to $L_{Q^{\mathrm{geom}}} \omega^\mathrm{geom} = 0$. The BV operator on functions on $V$ is $(-1)^{|e_i|}\tfrac 12 (\omega^\text{lin})^{ij}\partial_i  \partial_j$ which agrees with the conventions for vector spaces \cite[Section~3]{jpz:lagr_rel_and_quantum_Linfty}. When working in coordinates, we only use $\omega^\text{geom}$ in this paper.

\section{Hamiltonian flows and suspension}\label{app:Suspensionproof}
Let us prove that the relation $\Phi_H^\mathrm{rel}= \{ (\Phi^t_H(m), m, t, -H(m)) \mid m \in M, t \in \mathbb R \} $, defined in \eqref{eq:Lagrsuspensionlift}, is Lagrangian and compute the corresponding BV fiber integral along this relation, supplying the two missing calculations from Section \ref{sec:Hamflows}.

We will construct $\Phi_H^\mathrm{rel}\colon M \to M\times T^*[-1]\mathbb R$ as a composition:
% https://q.uiver.app/#q=WzAsNCxbMiwwLCJNXFx0aW1lcyBUXipbLTFdXFxtYXRoYmIgUiJdLFszLDAsIk1cXHRpbWVzIFReKlstMV1cXG1hdGhiYiBSIl0sWzEsMF0sWzAsMCwiTSJdLFswLDEsIlxcUHNpIl0sWzMsMCwiXFx7KG0sIG0sIHQsIDApIFxcbWlkIG0gXFxpbiBNLCB0XFxpbiBcXG1hdGhiYiBSXFx9Il1d
\[\begin{tikzcd}[column sep = 3.5em]
	M & {} & {} & {M\times T^*[-1]\mathbb R} & {M\times T^*[-1]\mathbb R}
	\arrow["{\Phi_0^\mathrm{rel} =\{ (m, m, t, 0) \mid m \in M, t\in \mathbb R\}}", from=1-1, to=1-4]
	\arrow["\operatorname{graph}\Psi^{-1}", from=1-4, to=1-5]
\end{tikzcd}\]
where $\Psi$ is the time $1$ flow on $M\times T^*[-1]\mathbb R$ of the following Hamiltonian
\begin{equation}
    t H \in \mathcal O(M\times T^*[-1]\mathbb R).
\end{equation}
Let us denote the time variable parametrizing the flow by $s$. The time $s$ flow of $tH$ is the time $ts$ flow of $H$; hence its integral curves are $m \mapsto \Phi_H^{ts}(m)$ on the $M$ factor. The coordinate $t$ stays constant. Finally, $\tau$ satisfies
\[ \frac{\partial \tau}{\partial s} = \{tH, \tau\} = H(\Phi_H^{ts}(m)). \]
Since the Hamiltonian flow preserves the degree $-1$ Hamiltonian also in odd symplectic geometry ($\{H, H\}=0$ for symmetry reasons), the change of $\tau$ is constant. The flow of $tH$ is all together
\[ (m, t, \tau) \mapsto (\Phi_H^{ts}(m), t, \tau + sH(m)) \]
At $s=1$, we get a symplectomorphism $\Psi$ that sends $(m, t, \tau)$ to $(\Phi^t_H(m), t, \tau + H(m))$ and thus the graph of its inverse can be written as 
\begin{equation}
    \operatorname{graph}\Psi^{-1} = \left( (\Phi^t_H(m), t, \tau, m, t, \tau-H(m)) \right).
\end{equation}
Composing with $\Phi_0^\mathrm{rel}$, we indeed get $\Phi_H^\mathrm{rel}$ from \eqref{eq:Lagrsuspensionlift}.
\medskip

Now for the BV pushforward of $\rho$ along $\Phi^\mathrm{rel}_H = \operatorname{graph}\Psi^{-1} \circ \Phi^\mathrm{rel}_0$. First, we compute the pushforward of $\rho\in \HalfDens(M)$ along $\Phi_0^\mathrm{rel}$, which is 
\[  \rho \tau  \sqrt{dt d\tau} \in \HalfDens(M\times T^*[-1]\mathbb R), \]
as one can easily check by pairing with a test half-density on $M\times T^*[-1]\mathbb R$
\begin{equation}\label{eq:referenceforintegralagainstest}\begin{split}
     \int_{M \times T^*[-1]\mathbb R}  \rho  \tau \sqrt{dt d\tau} \cdot (\mu_t + \tau \nu_t) \sqrt{dt d\tau} &= \int_M \int_\mathbb R\rho \mu_t dt  \\ & = \int_{\Phi^\mathrm{rel}_0\subset \overline{M}\times M \times T^*[-1]\mathbb R}  \rho \otimes (\mu_t + \tau \nu_t) \sqrt{dt d\tau} 
\end{split}
\end{equation}
As a second step, we use the simple relation
\[ 
\int_{\operatorname{graph}\Psi^{-1}} \sigma = \Psi^*\sigma.
\]
Our $\Psi$ is the $s=1$ flow of $tH$. We need to pullback $\rho\tau$:
\[ \Psi^* (\rho \tau) = \Psi^*(\rho) \Psi^*(\tau).\]
The pullback of the function $\tau$ is $\tau + H$; while the pullback of the half-density $\rho$ by $\Psi$ is equal to the pullback by $\Phi_H^{t}$; the flow in $T^*[-1]\mathbb R$ direction has no effect on $\rho$, by \eqref{eq:CartanMagicFormula}. Together, we get
\begin{equation}\label{eq:appendixBVpushforward}\int_{\Phi_H^\mathrm{rel}} \rho= (\Phi^t_H)^*(\rho) \cdot (\tau + H) \sqrt{dt d\tau}.\end{equation}
It would be interesting to generalize these calculations to time-dependent Hamiltonians $H_t$.

\printbibliography
\end{document}